\documentclass[article,aps,prd,9pt,notitlepage]{revtex4-2}

\usepackage{blindtext}
\usepackage{graphicx}
\usepackage{amsthm}
\usepackage{amsmath}
\usepackage{amssymb}
\usepackage{mathdots}
\usepackage{xcolor}
\usepackage{subcaption}
\usepackage{commath}
\usepackage{float}
\usepackage{fullpage}
\usepackage{appendix}
\usepackage{ragged2e}

\usepackage{bm}
\usepackage{amssymb}
\usepackage{graphicx,xcolor}
\usepackage{csquotes}
\usepackage{amsmath,graphicx}
\usepackage{braket}
\usepackage{bm}
\usepackage{tikz}
\usetikzlibrary{quantikz2}
\newtheorem{theorem}{Theorem}[section]
\newtheorem{corollary}{Corollary}[section]
\newtheorem{lemma}{Lemma}[section]
\newtheorem{corollary-for-lemma}[lemma]{Corollary}

\theoremstyle{remark}

\usetikzlibrary{quantikz2}
\theoremstyle{definition}
\newtheorem{definition}[theorem]{Definition}
\newcommand{\n}{\\ \nonumber}
\newcommand{\q}{\quad}
\newcommand{\I}{\mathbb{I}}
\newcommand{\qftn}[1]{\mathcal{F}_{#1}}
\newcommand{\qftdgn}[1]{\mathcal{F}_{#1}^\dagger}
\newcommand{\circeigv}[1]{\Gamma_{#1}}

\newcommand{\idenm}[1]{\mathbb{I}_{#1}}
\newcommand{\idenmnodim}{\mathbb{I}}

\newcommand{\hermpart}{\text{Herm}}
\newcommand{\skewhpart}{\text{Skew}}
\newcommand{\pauliz}{\sigma_Z}
\newcommand{\repart}{\mathfrak{Re}}
\newcommand{\impart}{\mathfrak{Im}}
\newcommand{\traceop}{\text{Tr}}
\newcommand{\bigoh}{\mathcal{O}}
\newcommand{\paulix}{\sigma_X}

\definecolor{revgreen}{rgb}{0.0,0.55,0.0}

\begin{document}

\title{Quantum algorithms for the exponentiation of Toeplitz matrices and applications in partial differential equations}
\author{ X. Guti\'errez$^{1,2}$, N. Mariella$^{3}$, J. Gonzalez-Conde$^{1,2,4}$, S. Zhuk$^{3}$, M. Sanz$^{1,2,5,6}$}
\affiliation{
$^{1}$ Department of Physical Chemistry, University of the Basque Country, UPV/EHU,Apartado 644, 48080 Bilbao, Spain\\
$^{2}$EHU Quantum Center, University of the Basque Country UPV/EHU, Apartado 644, 48080 Bilbao, Spain\\
$^{3}$ IBM Quantum, IBM Research Europe, Trinity Business School, Dublin, D02 F6N2 (Ireland) \\
$^{4}$ Quantum Mads, Calle Larrauri 1, Edificio A, piso 3, puerta 28, 48160 Derio, Spain\\
$^{5}$ Basque Center for Applied Mathematics (BCAM), Alameda de Mazarredo 14, 48009 Bilbao, Spain\\
$^{6}$IKERBASQUE, Basque Foundation for Science, Plaza Euskadi 5, 48009, Bilbao, Spain}
 \date\today
\begin{abstract}
We present quantum algorithms to approximate the exponential of banded Toeplitz matrices. These matrices have central importance in many PDE related problems, but their quantum implementation is hindered by their possibly large norms. Constructing directly the exponential we can circumvent this limitation. By relating the lower/upper shift operators to circulant and skew-circulant generators, which are diagonalised by the QFT, we construct (i) an LCU-based block encoding of banded Toeplitz matrices, (ii) an efficient, controllably truncated Pauli-string decomposition of the circulant eigenphases with a closed-form error bound, and (iii) a specialized QFT--Trotter product formula. As an application we build a block encoding of the propagator of the discretised heat equation with periodic, Dirichlet and Neumann boundary conditions, using a frequency cutoff.
\end{abstract}

\maketitle

\section{Introduction}
Partial Differential Equations (PDEs) are ubiquitous in science and engineering, describing phenomena ranging from fluid dynamics to quantum mechanics, electromagnetism or  financial modeling. A standard approach to their numerical resolution consists of spatial discretisation, which reduces the problem to the time evolution of a linear system $\partial_t \boldsymbol{u} = A\boldsymbol{u}$, whose solution is proportional to the matrix exponential, $\boldsymbol{u}(t) = e^{At}\boldsymbol{u}(0)$. The matrix $A$ arising from the discretisation typically exhibits a highly structured form, often of a {\it banded Toeplitz} or quasi Toeplitz matrix, such that its entries are constant over the diagonals, and only a few number of these diagonals, close to the main one, are nonzero \cite{HornMatrix,ToeplitzReview}. 

Quantum computers have emerged as a promising platform for solving such problems, offering the potential for an exponential advantage over its classical counterparts \cite{ChildsPDE, ArrazolaPDE,CostaPDE}. However, quantum algorithms operate through unitary transformations, whereas the matrix exponential $e^{A t}$ is generally nonunitary for those operators that appear in PDE discretisation. This limitation can be circumvented by the {\it block encoding} framework \cite{Qubitization,QSVT,Explicit,PowerBE}, which embeds a subnormalized matrix $A/\alpha$ as a block of a larger unitary $U_A$, where $\alpha \geq \norm{A}$. Combined with recently developed tools such as {\it Quantum Signal Processing} (QSP) \cite{QSP,Qubitization,Unification} and {\it Quantum Singular Value Transformation} (QSVT) \cite{QSVT}, block encodings enable efficient polynomial transformations of the encoded matrices. It is important to notice that the transformed matrices in these algorithms are necessarily unitaries in the case of QSP, while QSVT transforms the singular values of the encoded matrix, although this transformation only coincides with the usual polynomial of the matrix in the Hermitian case \cite{Survey}. A great effort has been developed in the last years to generalize these methods, resulting in extensions such as General QSP \cite{GQSP}, which lifts several of the previous restrictions over the allowed polynomial transformations. 

Despite the power of these tools, block-encoding based algorithms incur a complexity that scales with the subnormalization factor $\alpha$. This dependence becomes particularly problematic for differential operators.
A typical example of the kind of matrices arising in the study of PDE is the discretised Laplacian operator, $\Delta$ \cite{FDM}, whose norm diverges with the spatial grid $\delta_x$ as $\norm{\Delta} \sim 4/\delta_x^2$, making it ill-conditioned for generic block-encoding approaches. 

A natural strategy is therefore to exploit the algebraic structure of the matrices arising from PDE discretisations, as the exponentiation of (quasi) Toeplitz matrices. Considerable effort has been focused on solving PDEs \cite{ChildsPDE,Schrodingerisation,Carleman,HS-price,Poisson,Spectral} or on encoding sparse structured matrices by leveraging the structure and sparsity of the operators \cite{Explicit,Asymptotic,BELS,StructuredData,BEStructure}, including recent examples of the explicit block encoding of Laplacian operators \cite{BELap,Laplacian_BE}. Alternatively, analog algorithms have been proposed to solve non-linear and non-unitary PDEs \cite{gan2025provably}. Closer to the approach developed here, the {\it Quantum Fourier Transform}  (QFT) has also been used directly to solve PDEs by diagonalizing the discretized differential operator in frequency space, with applications in PDEs \cite{NISQ-wave,PDE-Fourier}. Nevertheless, the block encoding of the {\it exponentiation} of Toeplitz matrices has not been directly addressed. 

In this work we develop a quantum framework for directly implementing the block encoding of exponentials of banded Toeplitz matrices, circumventing the normalization bottleneck entirely. We notice that, although an arbitrary Toeplitz matrix is not generally diagonalisable, it can always be written as a polynomial of the upper and lower shift operators, $L_n$ and $L_n^\top$, respectively. Although these operators are nilpotent and thus non diagonalisable, they are block encoded by combinations of the circulant generator $C_n$ and the skew-circulant generator $N_n$, both of which are diagonalised by the QFT \cite{NielsenChuang}. 
We show that the diagonal eigenvalue matrix of the operators $C_n$ and $N_n$ can be well-approximated by a truncated Pauli-string decomposition with explicit error bounds. Thus, since the resulting Pauli strings commute, the exponentiation of these matrices can be reduced to the exponentiation of the truncated Pauli-strings, yielding an efficient implementation via single-qubit rotations. Finally, the decomposition of the Toeplitz matrix as a sum of a polynomial of $C_n$ and $N_n$ reduces the exponentiation to the {\it Trotter product} \cite{Hall2015,MolerVanLoan1978} of two QFT-diagonalisable unitaries.

When additional structure is known about the operator, the approach can be further specialized. As a concrete application we consider the one-dimensional heat equation, whose solution, after spatial discretisation, is given by the exponentiation of the discretised Laplacian $\Delta$. A direct block encoding of $e^{\Delta t}$ built by applying QSVT/QSP to a block encoding of the Laplacian is prohibitively costly. We can avoid this by working directly with the propagator: because its eigenvalues decay rapidly away from the low-frequency modes, we can truncate it with controlled error, removing the $\delta_x^{-2}$ scaling. The resulting truncated generator is non-unitary and can be implemented via Linear Combination of Hamiltonian Simulation (LCHS) \cite{Optimal-LCHS} as a combination of skew-Hermitian unitaries, each of which can be approximated by the exponentiation of the truncated Pauli-strings.

The paper is organized as follows. Section \ref{sec preliminary} introduces the notation and definitions that we will use throughout the work. In Section \ref{sec be toeplitz} we show how to represent a general non diagonalisable Toeplitz matrix in terms of diagonalisable operators as a block encoding. Section \ref{sec approximated be} develops an efficient implementation of an approximation to such a block encoding. In particular, we show how we can implement the exponential of each of its terms in an efficient manner. Section \ref{sec:qft-trotter} states the QFT-Trotter product formula for the exponentiation of Toeplitz matrices. In Section \ref{sec:heat-equation} we consider the case of the one-dimensional discretised heat equation. Section \ref{sec dicussion} summarizes our conclusions and discusses future directions.


\section{Preliminary}\label{sec preliminary}

We summarize here the main definitions that we are going to need for the rest of the work.

Given an arbitrary square matrix $A$, its Hermitian and skew-Hermitian components are defined as
\begin{align}
    \hermpart(A) \equiv \frac{A+A^\dagger}{2},\q \skewhpart(A) \equiv \frac{A-A^\dagger}{2},
\end{align}
such that $\hermpart(A)+\skewhpart(A) = A$. 

A {\it Toeplitz} matrix \cite{HornMatrix,ToeplitzReview} $V$ as a quantum operator acting on $n$ qubits is defined as
\begin{align}\label{def toeplitz}
    \bra{i}V\ket{j}=\begin{cases}\bra{0}V\ket{j-i}, & j\ge i,\\[2pt] \bra{i-j}V\ket{0}, & j<i,\end{cases}
\end{align}
such that all the elements of each of the diagonals have the same value. Equivalently, $\bra{i}V\ket{j}$ depends only on $j-i$. There are several kinds of specific Toeplitz matrices that are relevant for this work: in particular, we focus on {\it circulant}, {\it skew-circulant}, {\it lower shift} and {\it upper shift} matrices.

A {\it circulant} matrix \cite{HornMatrix,ToeplitzReview} $K$ as a quantum operator on $n$ qubits is defined by the relation
\begin{align}
    \bra{i}K\ket{j} = \bra{0}K\ket{(j-i) \bmod 2^n},
\end{align}
In a similar way, a {\it skew-circulant} (or {\it negacyclic}) matrix $K$ as a quantum operator acting on $n$ qubits is defined by
\begin{align}
    \bra{i}K\ket{j} = (\text{sgn}(j-i)+\delta_{i,j})\bra{0}K\ket{(j-i)\bmod 2^n}.
\end{align}

As a specific kind of circulant operators, we define the {\it cyclic permutation} matrix $C_n$ (also called  {\it circulant generator}) on $n$ qubits as
\begin{align}\label{cyclic permutation}
    \bra{i}C_n^k\ket{j} = \begin{cases}
        1,\q j-i = k(\bmod 2^n) \\0,\q \text{otherwise}.
    \end{cases}
\end{align}
Analogously, we define a specific kind of skew-circulant operators on $n$ qubits, called the {\it skew-circulant permutation} matrix $N_n$, also called {\it skew-circulant generator}. Given an integer $k$, its $k$-th power can be defined in terms of the cyclic permutation operator defined in Eq.(\ref{cyclic permutation}) as
\begin{align}\label{skew-circulant generator}
N_n^k = (-1)^q \left( \I - 2\sum_{i=1}^r \ket{2^n-i}\bra{2^n-i}\right)C_n^r,
\end{align}
with $q,r$ the uniquely specified integers satisfying $k=q2^n+r$, and $0\leq r \leq 2^n-1$.

In matrix form, these two operators are
\begin{align}
    C_n =
    \begin{pmatrix}
        0 & 1 &        &        &   \\
          & 0 & 1      &        &   \\
          &   & \ddots & \ddots &   \\
          &   &        & 0      & 1 \\
        1 &   &        &        & 0
    \end{pmatrix}\q N_n =
    \begin{pmatrix}
        0 & 1 &        &        &   \\
          & 0 & 1      &        &   \\
          &   & \ddots & \ddots &   \\
          &   &        & 0      & 1 \\
        -1 &   &        &        & 0
    \end{pmatrix},
\end{align}
that is, the skew-circulant $N_n$ differs from the circulant $C_n$ in that it has a sign change across the main diagonal. Notice that both the circulant and the skew-circulant operators are unitary.

Another kind of Toeplitz matrix is the {\it lower shift} matrix $L_n$ acting on $n$ qubits, defined by the relations $L_n^k\ket{i} = \ket{i+k}$ for nonnegative $i,k$ s.t. $i+k\leq 2^n-1$ and $\bra{0}L_n = 0$. By Eq.(\ref{def toeplitz}), we see it corresponds to a matrix with $1$s in the first subdiagonal, and $0$s everywhere else. Its transpose $L_n^\top = \paulix^{\otimes n}L_n\paulix^{\otimes n}$, where $\paulix = \begin{pmatrix}
    0&1\\1&0
\end{pmatrix}$ is the Pauli X operator, is called the {\it upper shift matrix}, and corresponds to the {\it Jordan block} \cite{HornMatrix} of the same order with eigenvalue $0$. These two matrices are nilpotent, such that $L_n^k = 0$ and $\left(L_n^\top \right)^k = 0$  for $k\geq 2^n$. We see that we can relate the circulant an skew-circulant generators with these upper and lower shift matrices.
\begin{lemma}\label{lemma:Cn-Nn}
    Let $n$ be a positive integer. Let $k$ be an integer and let $(q,r)$ be the unique pair of integers s.t. $k=q2^n+r$ and $0\leq r < 2^n$. Then
    \begin{align}\label{identity CN L}
        C_n^k &= \left( L_n^\top \right)^r +L_n^{2^n-r},\n
        C_{n+1}^k &= \paulix^q\otimes\left( L_n^\top \right)^r + \paulix^{q+1}\otimes L_n^{2^n-r},\n
        N_n^k &= (-1)^q\left( L_n^\top \right)^r +(-1)^{q+1}L_n^{2^n-r}.
    \end{align}
\end{lemma}
\begin{proof}
Since $C_n^{2^n}=\I$, $C_n^k\ket{j}=\ket{(j-r)\bmod 2^n}$ for every $0\le j<2^n$; for $j\ge r$ this coincides with $(L_n^\top)^r\ket{j}$ and for $j<r$ with $L_n^{2^n-r}\ket{j}$, and the two ranges are disjoint, proving the first identity. The second follows by tracking the wraparound parity on the extra qubit, and the third from $N_n^{2^n}=-\I$, which contributes the signs $(-1)^q$ and $(-1)^{q+1}$. 
\end{proof}

We anticipate an important consequence, which will play a central role in the construction of the specialized Trotter formula in Section \ref{sec:qft-trotter}. From Lemma.(\ref{lemma:Cn-Nn}) we directly see that
\begin{align}\label{eq:sum-Cn-Ln}
\frac{C_n^k+N_n^k}{2}=\frac{1+(-1)^q}{2}\left( L_n^\top \right)^r + \frac{1+(-1)^{q+1}}{2}L_n^{2^n-r}.
\end{align} 
This identity allows us to write a generally non-diagonalisable Toeplitz matrix as the sum of two diagonalisable operators, $C_n$ and $N_n$.

Given a $d$-degree polynomial with complex coefficients $p(x)=\sum_{k=0}^d \alpha_k x^k$, the matrix $p(L_n)$ is a lower triangular Toeplitz matrix, such that all the elements in the $k$-th subdiagonal have a value of $\alpha_k$. Similarly, $p(L_n^\top)$ is an upper triangular Toeplitz matrix. 

Let $\{\alpha_i\}$ be $2m+1$ real coefficients with $i\in[-m,m]$. We define the {\it banded Toeplitz operator} of bandwidth $m$ on $n$ qubits, assuming that $m = \text{poly}(n)$, by
\begin{align}\label{eq:toeplitz-poly}
    T = \alpha_0\I + \sum_{k=1}^m\left[\alpha_{-k}L_n^k + \alpha_k\left( L_n^\top\right)^k \right].
\end{align}
We restrict ourselves to banded matrices with equal bandwidth on both sides of the diagonal and half-bandwidth $m$, defined as the largest index $m$ such that $\alpha_m \neq 0$, although the general case, in which these are not equal, is equivalent.

Finally, a matrix closely related to the Toeplitz is the Hankel matrix. A Hankel matrix $H$ as an operator acting on $n$ qubits is a skew-diagonal Toeplitz matrix, defined by
\begin{align}
    \label{eq:hank-def}
    \bra{i}H\ket{j} =&
    \begin{cases}
        \bra{0}H\ket{j+i}, & i+j < 2^n,\\
        \bra{j+i-(2^n-1)}H\ket{2^n-1}, & i+j \ge 2^n.
    \end{cases}
\end{align}
We see that we can obtain a Hankel matrix from a related Toeplitz matrix $V$ as
\begin{align}\label{eq:Hankel-toep}
H = \paulix^{\otimes n}V = V^\prime \paulix^{\otimes n},
\end{align}
where $V$ and $V'$ are the two different Toeplitz matrices obtained by reversing the rows, resp.\ columns, of $H$.

These techniques rely on a block encoded matrix $A$ \cite{Qubitization,PowerBE}, embedded in a larger unitary operator, $U_A$.
Formally, we can give the following definition:
\begin{definition}[Block encoding]\label{def BE}
    \textit{Given a matrix $A$ acting on $n$ qubits, if we can find $\beta, \epsilon \in \mathbb{R}_+$ and a unitary matrix $U_A$ acting on $(m+n)$ qubits, such that}
    \begin{align}
         \norm{A-\beta(\bra{0}^{\otimes m}\otimes\mathbb{I})U_A(\ket{0}^{\otimes m}\otimes\mathbb{I})}_2\le\epsilon,
    \end{align}
    \textit{then $U_A$ is a $(\beta,m,\epsilon)$-block encoding of $A$. In particular, for $\epsilon = 0$ $U_A$ is a $(\beta,m)$-block encoding of $A$.}
\end{definition}
The normalizing factor $\beta$ is required in order for $U_A$ to be unitary, as any block of a unitary matrix must be normalized. Thus, a block encoding must satisfy that $\norm{A}/\beta \leq 1$. Notice that a unitary matrix is a $(1,0)$-block encoding of itself.

Finally, we might want to manipulate a block encoded matrix. In particular, given a unitary $U$ and a $d$-degree polynomial $p(x)$ with complex coefficients and both negative and positive powers, we are interested in the block encoding of the matrix $p(U)$. In order to do so, we will make use of Generalized Quantum Signal Processing (GQSP), which allows us to obtain such block encoding efficiently by $d$ controlled calls of the unitary $U$, a single ancilla qubit and $d+1$ single-qubit rotations.

\begin{theorem}[GQSP Theorem, from \cite{GQSP}]\label{thm:gqsp}
    $\forall d,k \in \mathbb{N}, \forall \vec{\theta},\vec{\phi}\in\mathbb{R}^{d+1}, \lambda \in \mathbb{R}$ and $k\leq d$ we have:
    \begin{align}\label{GQSP Negative}
        \begin{pmatrix}
            P^\prime (U) & * \\ Q^\prime(U) & * \end{pmatrix} = &\left[ \prod_{j=1}^k R(\theta_{d-k+j},\phi_{d-k+j},0) A^\prime\right] \times \left[ \prod_{j=1}^{d-k} R(\theta_j,\phi_j{\color{revgreen},}0)A\right]R(\theta_0,\phi_0,\lambda) 
    \end{align}
    If and only if:
    \begin{align}
        \begin{pmatrix}
            P(U) & * \\ Q(U) & * \end{pmatrix} = \left[ \prod_{j=1}^d R(\theta_j,\phi_j,0) A \right]R(\theta_0,\phi_0,\lambda)
    \end{align}
    For $P^\prime(U) = U^{-k}P(U)$ and $Q^\prime(U) = U^{-k}Q(U)$, where
    \begin{align}
        A^\prime = \begin{pmatrix}
            \mathbb{I} & 0 \\ 0 & U^\dagger
        \end{pmatrix}.
    \end{align}
\end{theorem}

The operators $R(\theta,\phi,\lambda) = \begin{pmatrix}
    e^{i(\lambda + \phi)}\cos(\theta) & e^{i\phi}\sin(\theta) \\ e^{i\lambda}\sin(\theta) & -\cos(\theta)
\end{pmatrix}$ are general $SU(2)$ rotations.

Linear Combination of Unitaries (LCU) \cite{LCU} is a quantum primitive that allows us to prepare block encodings of operators with the form $A = \sum_{j=0}^{J-1} c_j U_j$, where $U_j$ are unitary operators acting on $n$-qubits and $c_j$ are arbitrary coefficients. We assume that all the $c_j$s are real and positive, since the phase can be absorbed by the unitary operator. Assuming $J=2^a$, this is:
\begin{lemma}\label{lemma:lcu}
Let $U = \sum_{j=0}^{J-1}\ket{j}\bra{j}\otimes U_j$, called the {\it selector oracle}, and let $V$ acting as $V\ket{0^a} = \norm{c}_1^{-1/2}\sum_{j=0}^{J-1}\sqrt{c_j}\ket{j}$, where $\norm{c}_1 = \sum_j c_j$, called the {\it preparation oracle}. Let $W =\left(V^\dagger \otimes \I_n\right)U\left(V\otimes \I\right)$. Then, for a $n$-qubit state $\ket{\Psi}$,
\begin{align}
    W\ket{0}^a\ket{\Psi} = \frac{1}{\norm{c}_1}\ket{0^a}T\ket{\Psi} + \ket{\perp},
\end{align}
where $(\ket{0^a}\bra{0^a}\otimes \I_n)\ket{\perp}=0$.
\end{lemma}

We notice that $a$ only needs to be logarithmic in the number of terms. The selector oracle can be efficiently implementable if all the combined unitaries are powers of a given unitary \cite{QAlgCKS}.

\begin{lemma}\label{lemma:efficient-LCU}
 Let $U_j$ be a set of unitaries such that $U_j = U_0^j$. Then the select oracle $\sum_{j=0}^{J-1}\ket{j}\bra{j}\otimes U_j$ can be constructed with $\lceil \log(J)\rceil$ queries to controlled $U_j$'s.
\end{lemma}

Thus, we only need $\lceil \log(J)\rceil$ control qubits, such that each $0\leq l \leq \lceil \log(J)\rceil$ controls an application of $U_{2^l}$.

\section{Block-encoding of Toeplitz matrices}\label{sec be toeplitz}

We are interested in obtaining a block encoding for the banded Toeplitz defined in Eq.(\ref{eq:toeplitz-poly}). Because the upper and lower shift matrices are neither unitary nor Hermitian, we cannot implement the polynomial via GQSP or QSVT.
In order to treat them we first identify a relation between circulant matrices and lower and upper shift matrices, from Eq.(\ref{identity CN L}).
\begin{lemma}
    Let $C_{n+1}$ be a circulant generator acting on $n+1$ qubits as defined in Eq.(\ref{cyclic permutation}). Let $k$ be an integer and let $(q,r)$ be the unique pair of integers s.t. $k = q2^n+r$ and $0 \leq r < 2^n$. For a given single-qubit state $\ket{\phi}$ s.t. $\bra{\phi}\paulix\ket{\phi}\in \{ -1,0,1 \}$, then
    \begin{align}
        \left( \bra{\phi}\otimes \I \right)C_{n+1}^k\left( \ket{\phi}\otimes \I \right) = \bra{\phi}\paulix^q\ket{\phi}\left(L_n^\top \right)^r + \bra{\phi}\paulix^{q+1}\ket{\phi}L_n^{2^n-r}.
    \end{align}
\end{lemma}

In particular, considering the case for $\ket{\phi} = \ket{0}$, we see that the result is the Toeplitz matrix
\begin{align}\label{eq:Cn-be-T}
\left( \bra{0}\otimes \I \right)C_{n+1}^k\left( \ket{0}\otimes \I \right) = \frac{1+(-1)^q}{2}\left( L_n^\top \right)^r + \frac{1+(-1)^{q+1}}{2}L_n^{2^n-r}.
\end{align}
Note that the two terms are not simultaneously non-zero since $q$ and $q+1$ have complementary parities. We notice that, from the definition \ref{def BE}, this construction is a block encoding and can be extended to obtain the Toeplitz defined in Eq.(\ref{eq:toeplitz-poly}).  
Let ${\alpha_i}$ be $2m+1$ real coefficients with $i\in[-m,m]$. For $0\leq m \ll 2^n$, we have that
\begin{align}\label{eq:circulant-be-toeplitz}
    \left(\bra{0}\otimes \I \right)\left( \sum_{k=-m}^m\alpha_k C_{n+1}^k \right)\left(\ket{0}\otimes \I \right) = \alpha_0\I + \sum_{k=1}^m\left[\alpha_{-k}L_n^k + \alpha_k\left( L_n^\top\right)^k \right] = T.
\end{align}
Notice that, since the circulant generators are unitary matrices, the sum of the left hand side of Eq.(\ref{eq:circulant-be-toeplitz}), inside the parenthesis, is a LCU. 

We now turn our attention to the implementation of circulant matrices. These are unitarily diagonalisable under the discrete Fourier transform. In the context of quantum computing, this corresponds to the Quantum Fourier Transform (QFT) \cite{NielsenChuang}. Let $\qftn{n}$ denote the QFT operator, as
\begin{align}\label{qftn operator}
\qftn{n} = \frac{1}{\sqrt{2^n}}\sum_{i,j=0}^{2^n-1}\omega^{ij}\ket{i}\bra{j},
\end{align}
with $\omega=\exp(2\pi i/2^n)$. Thus, we can diagonalise the $n$-qubit circulant generator in Eq.(\ref{cyclic permutation}) as
\begin{align}\label{eq:Cn-diag}
C_n^k = \qftn{n}\Gamma_n^k\qftn{n}^\dagger,
\end{align}
for all $k \in \mathbb{Z}$, where $\Gamma_n$ is a diagonal matrix with the eigenvalues of the circulant generator $C_n$, given by
\begin{align}\label{eq:gamma-eig}
\Gamma_n = \sum_{k=0}^{2^n-1}\omega^k\ket{k}\bra{k}.
\end{align}
In order to relate it with its implementation, we notice that we can write it as the product of $n$ one-qubit operators.
\begin{lemma}\label{lemma:circ-eigvals-as-product}
The diagonal of the eigenvalues of Eq.(\ref{eq:gamma-eig}) exhibits a product operator structure, as
\begin{align}\label{eq:gamma-phase-gates}
\Gamma_n = \bigotimes_{l=1}^n\begin{pmatrix}
    1 & \\ & e^{2\pi i/2^l}
\end{pmatrix} = \bigotimes_{l=1}^n\begin{pmatrix}1&\\&-1\end{pmatrix}^{1/2^{l-1}},
\end{align}
that is, $\Gamma_n = \bigotimes_{l=1}^n P(2\pi/2^l)$, where $P(\theta)$ is a single qubit phase gate.
\end{lemma}
Thus, we can exactly implement a circulant generator via one call to the QFT operator $\qftn{n}$ and its adjoint, and $n$ single qubit phase gates each acting on one of the $n$ qubits encoding the operator.

\begin{figure}[H]
    \centering
\begin{quantikz}
    & \gate[4]{\qftn{n}}
    & \gate{Z^k}
    & \gate[4]{\qftn{n}^\dagger}
    &  \\
    && \gate{S^k} && \\
    && \gate{T^k} && \\
    &\qwbundle{} & \ \ldots \ &&
    \\
\end{quantikz} $\equiv$ \begin{quantikz}
    &\gate[2]{C_n^k}& \\
    &\qwbundle{} &
\end{quantikz}
\hfill
    \caption{Circuit implementation of the unitary $C_n^k$, using Eq.(\ref{eq:Cn-diag}) and Lemma \ref{lemma:circ-eigvals-as-product}.}
    \label{fig:Cnk}
\end{figure}

An alternative way to obtain the upper and lower shift matrices is a combination of the circulant and skew-circulant generators. In particular, from Eq.(\ref{eq:sum-Cn-Ln}) we notice that, for $k \geq 0$,
\begin{align}
\left(L_n^\top \right)^k = \frac{C_n^k+N_n^k}{2}, \q L_n^k = \frac{C_n^{-k}+N_n^{-k}}{2}.
\end{align}
This representation will be useful for the QFT-Trotter product formula in Section \ref{sec:qft-trotter}. We show in Appendix \ref{app skew-circulant diag} that there is a direct relation between the skew-circulant and the circulant generators that allows us to diagonalise the first one, as
\begin{align}\label{eq:Nn-diag}
N_n^k = \omega^{k/2}\sqrt{\Gamma_n}C_n^k\sqrt{\Gamma_n}^\dagger =  \omega^{k/2}\sqrt{\Gamma_n} \qftn{n}\Gamma_n^k\qftn{n}^\dagger\sqrt{\Gamma_n}^\dagger,
\end{align}
where $\omega^{k/2}$ is a global phase factor. Thus, we can implement in an equivalent way both the unitary operators $C_n^k$ and $N_n^k$.

\section{Approximated block encoding}\label{sec approximated be}

We note that each successive phase gate in Eq.(\ref{eq:gamma-phase-gates}) has a smaller phase and is thus closer to the identity operator. Thus, we can approximate the diagonal matrix by truncating the smaller phases and substituting them by identities. For a system of $n$ qubits, we define the approximation of degree $0\leq d \leq (n-1)$ for the eigenvalues $\Gamma_n$ of the circulant generator $C_n$ as 
\begin{align}\label{eq:gamma-approx}
\widetilde{\Gamma}_{d,n} \equiv \left( \bigotimes_{l=1}^{n-d} P(2\pi/2^l)\right)\otimes \I_2^{\otimes d} = \left( \sum_{k=0}^{2^{n-d}-1}\omega^{k2^d}\ket{k}\bra{k} \right)\otimes \I_2^{\otimes d}.
\end{align}
We see that we recover the exact expression for $d=0$, $\widetilde{\Gamma}_{0,n} =\Gamma_n$. We obtain the error of the approximation of integer powers of $\Gamma_n$, up to a global phase.
\begin{theorem}
    \label{thm:approx-circeigv-distance}
    Let $n, k, d, \delta$ be integers s.t.
    $n \ge 1$,
    $0 \le d \le n-1$,
    $\delta=n-d$,
    and $1 \le |k|\le 2^{\delta}$ with $k$ \underline{odd}.
    Then
    \begin{align}
        \label{eq:approx-circeigv-distance}
        \min_{\phi \in \mathbb{R}} \frac{1}{\sqrt{2^{n+1}}}\left\|
            \widetilde{\Gamma_{d,n}}^{k} - e^{\imath \phi} \circeigv{n}^k
        \right\|_F =&
        \sqrt{1-\frac{\sin(2^d \alpha)}{2^d \sin(\alpha)}}
        \approx
        \eta(\delta, k),
    \end{align}
    with
    \begin{align}
        \label{eq:approx-fun-eta}
        \eta(\delta, k)=&
        \sqrt{
        1-\frac{2^{\delta}}{\pi |k|}\sin\left(
            \frac{\pi |k|}{2^{\delta}}
        \right)},
    \end{align} 
    where $\alpha=\frac{2\pi |k|}{2^{n+1}}$.
    The approximation holds for large $n$ and $1\le |k| \ll 2^{n-d}$.
\end{theorem}

Proof in Appendix \ref{app:approx-circeigv}.
The restriction to odd $k$ can be lifted by exploiting the tensor-product structure of the powers of the circulant generator whose exponent is a power of two.
Since $C_n^{2^c}$ shifts the computational basis by $2^c$ positions,
it acts trivially on the $c$ least significant qubits, that is
\begin{align}
    \label{eq:Cn-pow2-tensor}
    C_n^{2^c} = C_{n-c}\otimes \I_2^{\otimes c},\quad&
    \Gamma_n^{2^c} = \I_2^{\otimes c}\otimes\Gamma_{n-c},
\end{align}
for $0\le c \le n-1$.
\begin{corollary}
    \label{cor:approx-circeigv-even}
    Let $n, d, \delta$ be as in Theorem \ref{thm:approx-circeigv-distance} and let $k=2^c k'$ with $k'$ odd, $0\le c\le \delta-1$ and $1\le |k| \le 2^{\delta}$.
    Then
    \begin{align}
        \label{eq:approx-circeigv-distance-even}
        \min_{\phi \in \mathbb{R}} \frac{1}{\sqrt{2^{n+1}}}\left\|
            \widetilde{\Gamma}_{d,n}^{k} - e^{\imath \phi} \circeigv{n}^k
        \right\|_F
        =&
        \min_{\phi \in \mathbb{R}} \frac{1}{\sqrt{2^{n-c+1}}}\left\|
            \widetilde{\Gamma}_{d,n-c}^{k'} - e^{\imath \phi} \circeigv{n-c}^{k'}
        \right\|_F\\
        =&
        \sqrt{1-\frac{\sin(2^d \alpha)}{2^d \sin(\alpha)}}
        \approx \eta(\delta, k),
    \end{align}
    with $\alpha=\frac{2\pi |k|}{2^{n+1}}$ as in Eq.(\ref{eq:approx-circeigv-distance}).
\end{corollary}
\begin{proof}
    By the tensor decomposition in Eq.(\ref{eq:Cn-pow2-tensor}),
    \begin{align}
        \widetilde{\Gamma}_{d,n}^{k} - e^{\imath\phi}\Gamma_n^{k} =&
        \I_2^{\otimes c}\otimes\big(\widetilde{\Gamma}_{d,n-c}^{k'} - e^{\imath\phi}\Gamma_{n-c}^{k'}\big)
    \end{align}
    whose Frobenius norm is $\sqrt{2^c}$ times that of the second factor, which gives the first equality.
    The second follows from Theorem \ref{thm:approx-circeigv-distance} applied to $(n-c, d, k')$
    noting that $2\pi|k'|/2^{n-c+1}=2\pi|k|/2^{n+1}=\alpha$
    and that $\eta(\delta-c, k')=\eta(\delta,k)$, since $\eta$ depends only on the ratio $|k|/2^{\delta}$.
\end{proof}

We obtain the logarithmic error when $n$ is large, that is
\begin{align}
    \label{eq:log-acirceigv-err}
    \log(\eta(\delta, k)) \sim& -\delta \log\left(\frac{2}{(\pi |k|)^{1/\delta}}\right),
\end{align}
as $\delta \to \infty$.
This shows that the approximation error $\eta(\delta, k)$ (with $k$ odd) in logarithmic scale, is eventually linear w.r.t. $\delta=n-d$.
Notably, the approximation depends on the difference $\delta=n-d$, therefore increasing the system size $n$ and the approximation degree $d$ at the same rate, results in the same approximation error.
Furthermore, due to the unitary invariance of the Frobenius norm and the eigendecomposition in \eqref{eq:Cn-diag}, the result extends naturally to the associated circulant matrices.

Although the expression in Eq.(\ref{eq:gamma-approx}) can directly be translated into a quantum circuit, thus allowing for a direct implementation of the circulant generator, we are actually interested in implementing its exponential. Unfortunately, $\exp\left[ \bigotimes_{l=1}^{n-d}P(2\pi/2^l) \right]$ leads to a highly entangled operator, with no efficient circuit implementation. We consider the approximation $\widetilde{\Gamma}_{d,n}$, which is diagonal and, as such,can be decomposed as a sum of combinations of $\pauliz$ and identity operators. As a matter of fact, in the case of the approximated operator we can do this decomposition efficiently.

\begin{theorem}\label{thm:overlap-zi-gamma-pows}
    For all integers $n\geq 2$ and $d$ s.t. $n-1\geq d \geq 1$, let
    \begin{align}
    \label{eq:local-pauli-string-iz}
    O_{\boldsymbol{p}}=&\left(\bigotimes_{\ell=1}^{n-d} \pauliz^{p_{\ell}}\right) \otimes \idenm{2}^{\otimes d},
    \end{align}
    with $\boldsymbol{p}$ being the vector of coefficients $p_l \in \{0,1\}$ determining a Pauli string in $\{\pauliz,\I_2 \}^{\otimes (n-d)}$. Then, for any odd integer $k\geq 1$
    \begin{align}
        \label{eq:overlap-zi-trace-details}
        \traceop\left(O_{\boldsymbol{p}} \circeigv{n}^k\right)
        =& 
        \frac{1-\xi_k}{2}
        \left(1 + \imath \cot\left(\frac{\pi k}{2^n}\right)\right)
        \prod_{\ell=1}^{n-d}\left(1 + (-1)^{p_{\ell}}\exp\left(\frac{2\pi \imath k}{2^{\ell}}\right)\right),
    \end{align}
    and $\xi_k=\exp(2\pi \imath \,k/2^{n-d})$.
    Notably, the trace vanishes when $p_1=0$, for all odd $k \ge 1$.
\end{theorem}

Proof in Appendix \ref{app:overlap}. 
The restriction to odd $k$ can again be removed by means of the tensor-product identities in Eq.(\ref{eq:Cn-pow2-tensor}).
\begin{corollary}
    \label{cor:overlap-zi-gamma-pows-even}
    Let $n$, $d$ and $O_{\boldsymbol{p}}$ be as in Theorem \ref{thm:overlap-zi-gamma-pows},
    and let $k\ge 1$ be an integer with $2^n \nmid k$, i.e. $k=2^c k'$ with $k'$ odd and $0\le c\le n-1$.
    Then Eq.(\ref{eq:overlap-zi-trace-details}) holds regardless of the parity of $k$.
    In particular, if $c \le n-d-1$, we recover the result from Eq.(\ref{eq:overlap-zi-trace-details}).
\end{corollary}
\begin{proof}
   We can rewrite Eq.(\ref{eq:local-pauli-string-iz}) as
    \begin{align}\label{eq:overlap-k-odd}
    O_{\boldsymbol{p}} =  \bigotimes_{l=1}^c\pauliz^{p_l} \otimes \left(\left( \bigotimes_{l=c+1}^{n-d}\pauliz^{p_l}\right)\otimes \I_2^{\otimes d}\right)\equiv O_{\boldsymbol{p}}^{(1)}\otimes O_{\boldsymbol{p}}^{(2)}.
    \end{align}
Inserting this relation in the left hand side of Eq.(\ref{eq:local-pauli-string-iz}) we obtain 
    \begin{align}
        \traceop\left( O_{\boldsymbol{p}}\circeigv{n}^k\right) = \traceop\left( O_{\boldsymbol{p}}^{(1)}\right) \traceop\left( O_{\boldsymbol{p}}^{(2)}\circeigv{n-c}^{k^\prime}\right),
    \end{align}
where
\begin{align}
    \traceop\left( O_{\boldsymbol{p}}^{(1)}\right) = \begin{cases}
        2^c \q \text{for } p_1 = p_2 = \cdots = p_c = 0\\ 0 \q \text{otherwise}
    \end{cases}.
\end{align}
Applying Theorem \ref{thm:overlap-zi-gamma-pows} to the second term in the right hand side of Eq.(\ref{eq:overlap-k-odd}) we obtain
\begin{align}
    \traceop\left( O_{\boldsymbol{p}}^{(2)}\circeigv{n-c}^{k^\prime}\right) = \frac{1-\xi^\prime_{k^\prime}}{2}\left( 1+i\cot\left( \frac{\pi k^\prime}{2^{n-c}}\right)\right) \prod_{l=1}^{n-c-d}\left[ 1+(-1)^{p_{l+c}}\exp\left(  \frac{2\pi k^\prime}{2^l}\right) \right],
\end{align}
where $\xi^\prime_{k^\prime} = \exp\left( 2\pi i k^\prime /2^{n-c-d}\right) = \exp\left( 2\pi i k2^{-c}/2^{n-c-d}\right) = \xi_k$ and $\cot(\pi k^\prime/2^{n-c}) = \cot(\pi k/2^n)$. Finally, we see that for $p_1 = p_2 = \cdots = p_c = 0$
\begin{align}
    \prod_{l=1}^{n-c-d}\left[ 1+(-1)^{p_{l+c}}\exp\left(  \frac{2\pi k^\prime}{2^l}\right) \right] = 2^c \prod_{l=1}^{n-d}\left[ 1+(-1)^{p_{l}}\exp\left(  \frac{2\pi k}{2^l}\right) \right].
\end{align}
\end{proof}

As a corollary, we can show that the coefficients in Eq.(\ref{eq:gamma-decomp}) decay rapidly. Assuming $k/2^n = \mathcal{O}(1/2^n)$ 
\begin{align}
    \frac{\abs{\traceop(O_{\boldsymbol{p}}\Gamma_n^k)}}{2^n} =&
    \gamma
    2^{n-d}\sin\left(\frac{\pi k}{2^{n-d}}\right)
    \cdot
    \prod_{\ell=1}^{n-d}\sqrt{\frac{1 + (-1)^{p_{\ell}}\cos(2\pi k/2^{\ell})}{2}} \approx
    \frac{1-\xi_k}{2\pi k}
    \prod_{\ell=1}^{n-d}\left(1 + (-1)^{p_{\ell}}\exp\left(\frac{2\pi \imath k}{2^{\ell}}\right)\right),
\end{align}
for some constant $\gamma > 0$. Proof in Appendix \ref{app:overlap}. The approximation follows from Lemma \ref{lemma:cot-i} when $n$ is sufficiently large. Notice that for $l \to \infty$, $\cos(2\pi k/2^l)\to 1$, and thus all Pauli factors get exponentially close to the identity operator $\I_2$ for large $l$. This justifies the expression in Eq.(\ref{eq:local-pauli-string-iz}). In particular, the constraint that bounds the scaling of $k/2^n$ can be interpreted as enforcing a banded structure on the corresponding block-encoded Toeplitz matrix.

In principle, in order to have a good enough approximation we would need to take into account each possible configuration of $\boldsymbol{p}$ in Eq.(\ref{eq:local-pauli-string-iz}), which are $2^{n-d}$, making the decomposition not very efficient. Nevertheless, in practice we can always avoid using all the subspace of the Pauli strings and consider only a subset $\Omega$ of the possible configurations, mainly because of two reasons. Firstly, we see from Eq.(\ref{eq:overlap-zi-trace-details}) that the coefficients decay rapidly for large $l$. Secondly, the pattern of which configuration $\boldsymbol{p}$ dominates is the same no matter the number of qubits, and can thus be known a priori, without the need to exhaust all the possible configurations.

Thus, we can approximate the diagonal operator $\Gamma_n$ into the subspace spanned by the Pauli strings, as
\begin{align}\label{eq:gamma-decomp}
    \Gamma_n^k \approx \widetilde{\Gamma}_{d,n}^k \approx \sum_{\boldsymbol{p}\in \Omega} \frac{\traceop(O_{\boldsymbol{p}}\Gamma_n^k)}{2^n}O_{\boldsymbol{p}}.
\end{align}
The expression in Eq.(\ref{eq:local-pauli-string-iz}) can be linearly extended to approximately decompose a general circulant matrix $\sum_{k=-m}^m \alpha_k C_n^k$, as the one in Eq.(\ref{eq:circulant-be-toeplitz}), in terms of the Pauli string. We consider here only part with odd powers, but the rest can be done equivalently. We can obtain the coefficients of the expansion as
\begin{align}
    g_{\boldsymbol{p}} \coloneqq&
    \frac{1}{2^n}
    \traceop\left(O_{\boldsymbol{p}}
        \sum_{k\ge 1,\, 2\nmid k}^m \alpha_k\circeigv{n}^k
    \right) =
    \sum_{k\ge 1,\, 2\nmid k} \alpha_k  \frac{\traceop\left(O_{\boldsymbol{p}} \circeigv{n}^k\right)}{2^n}.
\end{align}
Because now we have the diagonal matrix $\Gamma_n$ approximately decomposed in terms of commuting Pauli strings, its exponentiation can be written as a product of the exponential of each of the terms.
Thus, combining Theorem \ref{thm:overlap-zi-gamma-pows} and the eigendecomposition of the circulant matrices from Eq.(\ref{eq:Cn-diag}), we can easily implement the approximated exponentiation of the skew-hermitian part of the circulant matrix. This is
\begin{align}\label{eq:exp-diag-op}
    \exp\left(
        \sum_{k\ge 1,\, 2\nmid k}^m \alpha_k \skewhpart(C_n^k)
    \right)
    =&
    \qftn{n}
    \exp\left(
        \sum_{k\ge 1,\, 2\nmid k}^m \alpha_k\impart\left(\circeigv{n}^k\right)
    \right)
    \qftdgn{n}
    \approx
    \qftn{n}
    \left(
    \prod_{\boldsymbol{p} \in \Omega} \exp\left(
        \impart\left( g_{\boldsymbol{p}} O_{\boldsymbol{p}}
    \right)\right)\right)
    \qftdgn{n},
\end{align}
where the exponential of the Pauli strings can be efficiently implemented \cite{NielsenChuang}, as depicted in Figure \ref{fig:circuit-exp}.

Thus, we are able to approximately exponentiate diagonal matrices. But, as stated at the beginning, our goal was to obtain the exponential of Toeplitz matrices. We consider now the case in which we do not have any more information about the shape of the exponentiation of the Toeplitz, which for most applications is the general case.


\begin{figure}[H]
\centering
\begin{quantikz}[column sep=0.3cm]
\lstick[5]{$\ket{\Psi}$}
    & \gate[5]{\qftn{n}}
    & \gate{R_Z(g_{\boldsymbol{p_0}})} &\ctrl{2}&&&& \ctrl{2}&\ctrl{3}&&&& \ctrl{3}& \ \ldots \ & \gate[5]{\qftn{n}^\dagger}
    &  \\
 &&&&\ctrl{1}&&\ctrl{1}&&&\ctrl{2}&&\ctrl{2}&& \ \ldots \ && \\
&&&\targ{}& \targ{}&\gate{R_Z(g_{\boldsymbol{p_1}})}&\targ{}&\targ{}&&&&&& \ \ldots \ && \\
&&&&& &&&\targ{}&\targ{}&\gate{R_Z(g_{\boldsymbol{p_2}})}&\targ{}&\targ{}& \ \ldots \ && \\
&& \qwbundle{} &&&&&& \ \ldots \ &&&&&&&
\end{quantikz}
\caption{Structure of the circuit applying an approximation to the $\exp\left(\sum_{k\geq 1,2\nmid k}\alpha_k \skewhpart(C_n^k)\right)$ operator, Eq.(\ref{eq:exp-diag-op}), to a state $\ket{\Psi}$ on $n$ qubits. The exponentiation of the sum of commuting Pauli strings is reduced to the product of Pauli-string rotations $e^{g_{\boldsymbol p}O_{\boldsymbol p}}$, with $R_Z(\theta)=e^{\imath\theta\pauliz}$. Note that for strings of weight $>1$ (e.g.\ $\pauliz\otimes\pauliz$) the rotation is \emph{not} a tensor product of single-qubit gates: each layer is implemented by the standard CNOT-ladder conjugation of one $R_Z$ \cite{NielsenChuang}. The number of layers depends on the degree $d$ of the truncation (see Theorem \ref{thm:approx-circeigv-distance}).}\label{fig:circuit-exp}
\end{figure}

\section{The QFT-Trotter product formula}\label{sec:qft-trotter}

An arbitrary Toeplitz matrix is not necessarily diagonalisable. Nevertheless, we can always write them as sums of circulant and block-circulant matrices, which are  always diagonalisable. As an illustration, consider
\begin{align}
    \label{eq:simple-toeplitz-from-circ-and-negac}
    T =&
    \begin{pmatrix}
        0 & 1 &        &        &   \\
          & 0 & 1      &        &   \\
          &   & \ddots & \ddots &   \\
          &   &        & 0      & 1 \\
        0 &   &        &        & 0
    \end{pmatrix}
    =
    \frac{1}{2}
    \begin{pmatrix}
        0 & 1 &        &        &   \\
          & 0 & 1      &        &   \\
          &   & \ddots & \ddots &   \\
          &   &        & 0      & 1 \\
        1 &   &        &        & 0
    \end{pmatrix}
    +
    \frac{1}{2}
    \begin{pmatrix}
        0  & 1 &        &        &   \\
           & 0 & 1      &        &   \\
           &   & \ddots & \ddots &   \\
           &   &        & 0      & 1 \\
        -1 &   &        &        & 0
    \end{pmatrix} = \frac{1}{2}C_n + \frac{1}{2}N_n,
\end{align}
where $C_n$ and $N_n$ are the cyclic permutation and negacylic permutation matrices defined in Eqs.(\ref{cyclic permutation}) and (\ref{skew-circulant generator}), respectively. As we saw in the previous section, although we cannot  directly implement an exponentiation of the Toeplitz $T$ in Eq.(\ref{eq:simple-toeplitz-from-circ-and-negac}) since it is not diagonalisable,  we can efficiently implement an approximation to the exponentiation of each of the matrices in the decomposition, $C_n$ and $N_n$, which are diagonalisable (see Eqs.(\ref{eq:Cn-diag}) and (\ref{eq:Nn-diag})). Thus, we can apply the Lie-Trotter product formula \cite{Hall2015,MolerVanLoan1978} to leverage the results from the previous section, and obtain an approximation to the exponentiation of a general banded Toeplitz matrix. Formally, we can write
\begin{lemma}\label{lemma:banded-T-D1-D2}
    Let $T$ be the banded Toeplitz operator defined in Eq.(\ref{eq:toeplitz-poly}), given by the coefficients $\{\alpha_i\}$ and bandwidth $m\ll 2^n$. Let $D_1$ and $D_2$ be diagonal operators on $n$-qubits, defined as
\begin{align}\label{eq:D1-D2}
    D_1 \equiv \sum_{k=-m}^m\alpha_k\Gamma_n^k,\q D_2\equiv \sum_{k=-m}^m \alpha_k \omega^{k/2}\Gamma_n^k.
\end{align}
Then,
\begin{align}\label{eq:new-trotter-lemma-first}
    T = \frac{1}{2}\qftn{n}D_1\qftn{n}^\dagger +\frac{1}{2}\sqrt{\Gamma_n}\qftn{n}D_2\qftn{n}^\dagger \sqrt{\Gamma_n}^\dagger,
\end{align}
where $\omega = \exp\left( 2\pi i/2^n\right)$.
\end{lemma}
Proof in Appendix \ref{app:banded-toeplitz}. Notice that the left hand side is not necessarily a normal operator, whereas the right hand side is a linear combination of normal operators whose basis do not commute.

 We consider an implementation by the  approximated eigenvalues of the circulant generator, given in Eq.(\ref{eq:gamma-approx}), such that we define the approximated version of the diagonal operators in Eq.(\ref{eq:D1-D2}) as
\begin{align}\label{eq:D1-D2-app}
    \widetilde{D}_1 \equiv \sum_{k=-m}^m\alpha_k\widetilde{\Gamma}_{d,n}^k,\q \widetilde{D}_2\equiv \sum_{k=-m}^m \alpha_k \omega^{k/2}\widetilde{\Gamma}_{d,n}^k.
\end{align}

We can now apply the Lie-Trotter product formula.
\begin{theorem}[QFT-Trotter product formula]\label{thm:qft-trotter}
    Let
\begin{align}
    T = \sum_{k=1}^m \left[ \alpha_{-k}L_n^k +\alpha_k\left( L_n^\top \right)^k\right] + \alpha_0\I
\end{align}
be a banded Toeplitz operator. Let $n, d, \delta$ be integers s.t. $n \ge 1$, $0 \le d \le n-1$, $\delta=n-d$. Then
\begin{align}
   \norm{e^{Tt} - \left( \widetilde{A}^{t/c}\widetilde{B}^{t/c}\right)^c} = \mathcal{O}\left(\frac{t^2}{c}\norm{\alpha}_1^2 \right) +  \mathcal{O}\left( \frac{\pi t m}{2^\delta}\norm{\alpha}_1\right),
\end{align}
with
\begin{align}\label{eq:A-B}
    \widetilde{A}=&\qftn{n}\exp\left( \frac{\widetilde{D}_1}{2}\right)\qftn{n}^\dagger,\q
    \widetilde{B}=\sqrt{\Gamma_n}\qftn{n}\exp\left( \frac{\widetilde{D}_2}{2}\right)\qftn{n}^\dagger\sqrt{\Gamma_n}^\dagger,
\end{align}
where $\widetilde{D}_1$ and $\widetilde{D}_2$ are given by Eq.(\ref{eq:D1-D2-app}), $\norm{\alpha}_1 \equiv \sum_{k=-m}^m\abs{\alpha_k}$, and we have assumed that $\max_j \repart \widetilde{D}_{1,jj} \leq 0$ and $\max_j \repart \widetilde{D}_{2,jj} \leq 0$.
\end{theorem}
Proof in Appendix \ref{app:qft-trotter}. 
Because the arguments of $\exp(\cdot)$ in Eq.(\ref{eq:A-B}) are diagonal, we can implement their approximate exponentiation using the results of the previous section. This allows us to obtain an approximation to $\exp(T)$, where $T$ is a banded Toeplitz matrix, in the general case, in which we only know information about the coefficients $\alpha_i$. 


When additional structure of the banded Toeplitz matrix is known the generic Trotter approach can be substantially improved. In particular, we consider the exponentiation of an operator arising from the discretisation of a differential operator, that can be well approximated by an operator with a much narrower spectral support.  

\section{Application to the discretised heat equation}\label{sec:heat-equation}

As an example of such encoding, we consider the case of the one-dimensional heat equation,$\frac{\partial \boldsymbol{u}}{\partial t} = \frac{\partial^2 \boldsymbol{u}}{\partial x^2}$. Discretising the spatial coordinate by a grid of separation $\delta_x$ and a total length $L$, the equation can be rewritten as
\begin{align}\label{eq:discr-heat}
    \frac{\partial}{\partial t}\boldsymbol{u}(t) = \Delta \boldsymbol{u}(t),
\end{align}
where $\Delta$ is the discretised Laplace operator. The matrix form depends on the boundary conditions; we consider periodic (p), Dirichlet (D) and Neumann (N) boundary conditions. The corresponding matrices are
\begin{align}\label{eq:Lap-mat}
\Delta_p = \frac{1}{\delta_x^2}\begin{pmatrix}
    -2 & 1 & 0 & \hdots & 1 \\
    1 & -2 & 1 & \hdots & 0 \\
    \vdots \\ \\
    1 & 0 & \hdots & 1 & -2
\end{pmatrix}, \q \Delta_D = \frac{1}{\delta_x^2}\begin{pmatrix}
    -2 & 1 & 0 & \hdots & 0 \\
    1 & -2 & 1 & \hdots & 0 \\
    \vdots \\ \\
    0 &  & \hdots & 1 & -2
\end{pmatrix}, \q \Delta_N = \frac{1}{\delta_x^2}\begin{pmatrix}
    -1 & 1 & 0 & \hdots & 0 \\
    1 & -2 & 1 & \hdots & 0 \\
    \vdots \\
    &  \hdots &1 & -2 & 1\\
    0 &  & \hdots & 1 & -1
\end{pmatrix},
\end{align}
where $\delta_x = L/2^n$.
For $\boldsymbol{b}(t)=0$, the solution to Eq.(\ref{eq:discr-heat}) is $\boldsymbol{u}(t) = \exp\left( \Delta_i t\right)\boldsymbol{u}(0)$. We could thus implement the block encoding of these matrices with the methods presented of the previous section or different methods \cite{BELap,Laplacian_BE} and then apply QSP/QSVT. Nevertheless, the prohibitively large value of the norm ($\norm{\Delta_i}\sim 4/\delta_x^2$) precludes us from using these methods, as their cost increases drastically. We can nevertheless leverage our knowledge about the shape of $\exp(T)$.

Notice that discretised Laplacian with periodic boundary conditions, $\Delta_p$ in Eq.(\ref{eq:Lap-mat}), is a circulant matrix, and is thus diagonalised by the QFT. For an operator acting on $n$ qubits, the eigenvalues are
\begin{align}\label{eq:heat-ev}
    \lambda_k = \delta_x^{-2}(-2+2\cos(2\pi k/2^n)) = -\delta_x^{-2}4\sin^2(\pi k/2^n), \q \text{with}\q k=0,1,\dots,2^n-1.
\end{align}

\subsection{Filtering high frequencies}


For large enough $t$, the exponential of the eigenvalues in Eq.(\ref{eq:heat-ev}) decays as $\exp\left( -4t\sin^2(\pi k/2^n)/\delta_x^2\right)$, so we can filter any large eigenvalues with a controlled error. In particular, we keep only frequencies close to $k=0$ or $k=2^n$, where the eigenvalues have their minumum absolute value. Considering first periodic boundary conditions, we can write the Laplacian as $\Delta_p = \frac{1}{\delta_x^2}\left(C_n^\top -2\I +C_n \right)$. We can write the truncated Laplacian using Eq.(\ref{eq:Cn-diag}), such that 
\begin{align}\label{eq:diag-cutoff}
\widetilde{\Delta_p}^{(k_c)} =&\qftn{n}\left(\sum_{\substack{0\leq k \leq k_c \\ 2^n-k_c \leq k <  2^n}} -\frac{4}{\delta_x^2}\sin^2(\pi k/2^n)\ket{k}\bra{k} \right)\qftn{n}^\dagger,
\end{align}
where $k_c$ is the cutoff. For a given normalized initial state $\boldsymbol{u}(0)$, the cutoff error is 
\begin{align}\label{eq:cutoff-err}
\norm{e^{\Delta_p t}\boldsymbol{u}(0)-e^{\widetilde{\Delta_p}^{(k_c)} t}\boldsymbol{u}(0)} \leq&\norm{\qftn{n}\left( \sum_{k=0}^{2^n-1}e^{\lambda_k t}\ket{k}\bra{k}-\sum_{\substack{0\leq k \leq k_c \\ 2^n-k_c \leq k <  2^n}}e^{\lambda_k t}\ket{k}\bra{k}
\right)\qftn{n}^\dagger}\norm{\boldsymbol{u}(0) }
\leq \n
\leq&\max_{k_c<k<2^n-k_c}\abs{e^{\lambda_k t}} = \exp\left( -\frac{4\sin^2(\pi (k_c+1)/2^n)}{\delta_x^2}t\right) \approx \exp\left(-\frac{4k_c^2\pi^2}{L^2}t\right),
\end{align}
where we have used that $\delta_x = L/2^n$. Notice that this approximation holds near $k=0$ and $k = 2^n$, since $\sin(x) = \sin(\pi-x)$, and thus $\sin(\pi(2^n-k)/2^n) = \sin(\pi - \pi k /2^n) = \sin(\pi k /2^n)$.

Imposing that the error of Eq.(\ref{eq:cutoff-err}) remains below a threshold $\epsilon$, we obtain a cutoff
\begin{align}\label{eq:kc}
k_c = \mathcal{O}\left( L\sqrt{\frac{1}{t}\log\left( \frac{1}{\epsilon}\right)}\right).
\end{align}

Assuming that $k_c = 2^m-1$ for some integer $m$, we can define the projector over the non-zeroed states as $\Pi_m^{(2)} \equiv \Pi_m \oplus 0_{n-2m}\oplus \Pi_m$, where $\Pi_m = \sum_{k=0}^{2^m-1}\ket{k}\bra{k}$.
We can rewrite Eq.(\ref{eq:diag-cutoff}) using Eq.(\ref{eq:gamma-eig}) and defining
\begin{align}
\Gamma_m^{(2)} \equiv \Gamma_m \oplus 0_{n-2m}\oplus w^{-2^m}\Gamma_m
\end{align}
as
\begin{align}
    \widetilde{\Delta_p}^{(k_c)} = \qftn{n}\left[ \frac{1}{\delta_x^2}\left(\Gamma_m^{(2)} \right)^{-1}-\frac{2}{\delta_x^2}\Pi_m^{(2)}+\frac{1}{\delta_x^2}\Gamma_m^{(2)}\right]\qftn{n}^\dagger.
\end{align}

The heat propagator with periodic boundary conditions is thus obtained as
\begin{align}\label{eq:block-diag}
\exp\left( \widetilde{\Delta_p}^{(k_c)}t \right)=\qftn{n}\left( \exp\left[\frac{1}{\delta_x^2}\Gamma_m^{-1}t-\frac{2}{\delta_x^2}\I_mt+ \frac{1}{\delta_x^2}\Gamma_mt\right]\oplus \I_{n-(m+1)}\oplus \right.\n
\left. \oplus \exp\left[\frac{w^{2^m}}{\delta_x^2}\Gamma_m^{-1}t-\frac{2}{\delta_x^2}\I_mt+ \frac{w^{-2^m}}{\delta_x^2}\Gamma_mt\right] \right)\qftn{n}^\dagger,
\end{align}
where we have used that $\exp[0(\I_{n-(m+1)})] = \I_{n-(m+1)}$.

In Eq.(\ref{eq:block-diag}) we have a block diagonal matrix composed of three diagonals; the two exponential of the $\Gamma_m$, corresponding to the frequencies that we are keeping, and the identity corresponding to the zeroed frequencies. We name 
\begin{align}
U_{lo} &= \exp\left[X_{lo}t\right] \equiv \exp\left[\frac{1}{\delta_x^2}\Gamma_m^{-1}t-\frac{2}{\delta_x^2}\I_mt+ \frac{1}{\delta_x^2}\Gamma_mt\right],\n U_{hi} &= \exp\left[ X_{hi}t\right] \equiv \exp\left[\frac{w^{2^m}}{\delta_x^2}\Gamma_m^{-1}t-\frac{2}{\delta_x^2}\I_mt+ \frac{w^{-2^m}}{\delta_x^2}\Gamma_mt\right],
\end{align}
such that the block diagonal matrix is $\exp\left( \widetilde{\Delta_p}^{(k_c)}\right) = \qftn{n}\left(U_{lo}\oplus \I_{n-(m+1)}\oplus U_{hi}\right)\qftn{n}^\dagger$. Notice that
\begin{align}
\norm{X_{lo}} = \norm{X_{hi}} = \max_k \abs{ \lambda_k} = \frac{4 k_c^2 \pi^2}{L^2} = \frac{4 \pi^2\log(1/\epsilon)}{t},
\end{align}
such that there is no $\delta_x^{-1}$ and thus no $2^n$ dependence.

Because $U_{lo}$ and $U_{hi}$ are non-unitary evolutions, we cannot apply directly the results from Section \ref{sec approximated be}. We can use Linear Combination of Hamiltonian Simulation (LCHS) \cite{Optimal-LCHS} in order to capture this non-unitary dynamics.Given a matrix $A$, LCHS allows us to construct $\exp(-At)$ with $\mathcal{O}(\norm{A} t\log(1/\epsilon))$ calls to a Hamiltonian simulation oracle, given that $A+A^\dagger \succeq 0$. In this case, because we want to construct $\exp(X_{lo/hi}t)$, we need that $X_{lo/hi}+X_{lo/hi}^\dagger \preceq 0$, which they already satisfy. The general transformation considers $A = L+iH$, where $L = \hermpart(A)$ and $H = -i\skewhpart(A)$. In this case, $X_{lo/hi} = \hermpart(X_{lo/hi})$, and there is no skew-Hermitian component. The transformation is
\begin{align}
U_{lo/hi}= e^{X_{lo/hi}t} = \frac{1}{\sqrt{2\pi}} \int_{\mathbb{R}} \hat{f}(k) e^{-i k X_{lo/hi}t}dk \approx \sum_{j=1}^M w_j \left(e^{i\Delta_k X_{lo/hi}t} \right)^{j},
\end{align}
where $\hat{f}(k)$ is the kernel function , $w_j = \hat{f}(k_j)\Delta_k/\sqrt{2\pi}$ and $k_j = j\Delta_k$. Each of the unitaries $\exp[i \Delta_k X_{lo/hi}t]$ can in turn be implemented using Section \ref{sec approximated be} since $i\Delta_k X_{lo/hi}$ is skew-Hermitian. The identity can be taken out of the exponential as a constant phase $e^{-2i\Delta_k t/\delta_x^2}$, and we are left with a skew-hermitian matrix, similar to the one in Eq.(\ref{eq:exp-diag-op}).

Thus, $U_{lo/hi}$ can be implemented with $\mathcal{O}(\log^2(1/\epsilon))$ calls to the implementation of each of the unitaries using Section \ref{sec approximated be}. 

We can construct the block diagonal matrix using controlled applications of each of the operators $U_{lo/hi}$. Notice that, from LCHS, this is embedded in a larger system, using $\mathcal{O}(\log M)$ ancilla qubits. Naming the projectors $\Pi_0 = \ket{0}\bra{0}^{\otimes n-(m+1)}$ and $\Pi_1 = \ket{1}\bra{1}^{\otimes n-(m+1)}$ we can construct the controlled operators
\begin{align}
C_0 \equiv \Pi_0\otimes U_{lo} + \left( \I_n-\Pi_0\right),\q  C_1 \equiv \Pi_1\otimes U_{hi} + \left( \I_n-\Pi_1\right),
\end{align}
such that 
\begin{align}
C_0C_1 = \Pi_0\otimes U_{lo} + \Pi_1 \otimes U_{hi} + \left(\I_n-\Pi_0-\Pi_1\right) = U_{lo} \oplus I_{n-(m+1)}\oplus U_{hi},
\end{align}
and thus 
\begin{align}
\exp\left( \widetilde{\Delta_p}^{(k_c)}t\right) = \qftn{n}C_0C_1\qftn{n}^\dagger.
\end{align}

{\it Non-periodic boundary conditions:} the non-periodic Laplacians are not diagonalized via the QFT, but via the Discrete Cosine Transformation (DCT) and the Discrete Sine Transformation (DST), for Neumann and Dirichlet boundary conditions, respectively. Because these are not unitary operations, they have to be embedded in a $n+1$ qubit system, and the Laplacians are recovered as a block encoding.

We consider Neumann boundary conditions, whose Laplacian is diagonalized via DCT-II \cite{DCT}. The Dirichlet case of Eq.(\ref{eq:Lap-mat}) can equivalently be done using DST-I. We can use
\begin{align}\label{eq:DCT-DST-diag}
    V_{n+1}^\dagger \Delta_p^{(n+1)} V_{n+1} =  V_{n+1}^\dagger \qftn{n+1} \sum_{k=0}^{2^{n+1}} \lambda_k \ket{k}\bra{k}\qftn{n+1}^\dagger V_{n+1} = \Delta_N \oplus \Delta_D,
\end{align}
where $\Delta_p^{(n+1)}$ is the Laplacian with periodic boundary conditions acting on $n+1$ qubits, $\lambda_k$ are the eigenvalues of Eq.(\ref{eq:heat-ev}) for an operator acting on $n+1$ qubits, and
\begin{align}
        V_{n+1} =&\frac{1}{\sqrt{2}}
    \begin{pmatrix}
         1 && 1       &&   \\
            & \ddots  && \ddots &   \\
           && 1     &&    1 \\
           && 1 && -1 \\
           &  \iddots &&  \iddots & \\
           1 && -1 &&
    \end{pmatrix}.
\end{align}
 Proof in Appendix \ref{app DCT DST}. $V_{n+1}$ can be implemented via a reflection and two Haddamard gates, and thus the cost is dominated by the cost of the implementation of QFT, which is $\mathcal{O}(n^2)$ \cite{DCT-quantum,Fast-qrt}. 

Because the matrix in Eq.(\ref{eq:DCT-DST-diag}) is block diagonal, its exponential is the exponential of each of the blocks, and thus we can obtain for Neumann boundary conditions (and similarly for Dirichlet)
\begin{align}
e^{\Delta_N t} = \left(\bra{0}\otimes \I_n \right)V_{n+1}^\dagger \qftn{n+1} \sum_{k=0}^{2^{n+1}} e^{\lambda_k t} \ket{k}\bra{k}\qftn{n+1}^\dagger V_{n+1} \left(\ket{0}\otimes \I_n\right),
\end{align}
and thus we can apply the same frequency cutoff as in Eq.(\ref{eq:diag-cutoff}).

\section{Discussion and conclusions}\label{sec dicussion}

In this work we have developed a quantum framework for a straightforward implementation of a block encoding of exponentials of banded Toeplitz matrices, with direct applications to quantum resolution of Partial Differential Equations (PDEs). The key insight is that, although banded Toeplitz matrices are generally non-diagonalisable, they can be obtained as the block encoding of a polynomial of the circulant generator $C_n$ (see Eq.(\ref{eq:Cn-be-T})), which is unitary and diagonalisable. More generally, we showed that a general Toeplitz can be decomposed in terms of $C_n$ and the skew-circulant generator $N_n$, both of which are diagonalised by the Quantum Fourier Transform (QFT), Eqs.(\ref{eq:Cn-diag}) and (\ref{eq:Nn-diag}). 

We showed that this QFT-diagonalisation of both $C_n$ and $N_n$ can be represented in terms of the same diagonal matrix containing their eigenvalues, $\Gamma_n$. We derived an efficient approximation for this diagonal operator based on truncated Pauli-Z string decomposition with a closed form Frobenius-norm error bound, Theorem \ref{thm:overlap-zi-gamma-pows}. This decomposition depends only on the truncation depth $\delta=n-d$ and not directly on the system size $n$. Since the resulting Pauli strings commute, an approximation to the exponentials of $C_n$ and $N_n$ can be efficiently implemented through the QFT diagonalisation followed by single-qubit rotations,  Eq.(\ref{eq:exp-diag-op}).

Building on this construction, we showed that by applying the Lie-Trotter product formula to the circulant/skew-circulant decomposition of the banded Toeplitz $T$ we can obtain an approximation to $e^{T t}$ with an error bounded by $\mathcal{O}((t^2/c))\abs{\alpha}_1^2)+\mathcal{O}\left( \frac{\pi t m}{2^\delta}\norm{\alpha}_1\right)$, where $c$ is the number of Trotter steps, $\norm{\alpha}_1^2 = \left(\sum_k \abs{\alpha_k}\right)^2$, determined by the coefficients $\{\alpha_k\}$ of the Toeplitz matrix (Theorem \ref{thm:qft-trotter}), and $m$ is the bandwidth of the matrix $T$.

As a concrete application, we have considered the one-dimensional heat equation, whose solution is the exponentiation of the discretised Laplacian $\Delta$. Since $\Delta$ is circulant under periodic boundary conditions, it is diagonalised by the QFT. Similarly, considering Dirichlet and Neumann boundary conditions, the operator is diagonalized via the discrete sine and cosine transform, respectively, with a similar cost to the QFT. We showed that, since the eigenvalues decay rapidly away from $k = 0$, the propagator can be truncated to a narrow low-frequency band with an error controlled by the frequency cutoff $k_c = \mathcal{O}(L\sqrt{(1/t)\log(1/\epsilon)}$. Because the truncated low-frequency block generates non-unitary evolutions, we can implement them using Linear Combination of Hamiltonian Simulation (LCHS), decomposing each block into a linear combination of skew-Hermitian unitaries implemented through the truncated Pauli-string construction of Section \ref{sec approximated be}. Thus, we eliminate the normalization bottleneck incurred when applying QSP/QSVT protocols direcly to a block encoding of $\Delta$, whose norm scales as $\mathcal{O}(1/\delta_x^2)$, thus increasing exponentially with the number of qubits.

Prior quantum algorithms for Toeplitz linear systems \cite{Asymptotic,BELS} and structured-matrix block encodings \cite{Explicit,StructuredData} address the problem of block encoding Toeplitz operators, but do not directly construct operators with the form $e^{T t}$, as the ones arising from PDEs. Recent explicit encodings of Laplacian operators \cite{BELap,Laplacian_BE} provide efficient circuits for $\Delta$ itself, but at the cost of a high subnormalisation which prevents applying QSP or QSVT related methods. 

Our results avoid this bottleneck by establishing a general methodology for exploiting the algebraic structure of banded Toeplitz operators. By directly addressing the implementation of the exponentiation of those matrices, we can overcome the prohibitive cost of the subnormalization factors, achieving a $\mathcal{O}(1)$ subnormalization. Compared to general PDE solvers \cite{ChildsPDE,ArrazolaPDE,CostaPDE}, the method trades generality for a better scaling in the regime where the bandwidth of the Toeplitz operator is small compared to the system size, $m \ll 2^n$, and where the norm of the operator is large. These conditions generally arise in the discretisation of differential operators.

Several limitations and open questions remain, which we state explicitly. First, the Trotter bound of Theorem~\ref{thm:qft-trotter} is stated up to constants that, for non-contractive coefficient sets, include a factor exponential in $t\norm{\alpha}_1$; a tight statement for general banded Toeplitz operators is open. Second, the selection of the dominant Pauli-string subset $\Omega$ in Sec.~\ref{sec approximated be} is currently heuristic, based on the observed $n$-independence of the dominant pattern; a rigorous a-priori bound on $|\Omega|$ versus the target error would strengthen the result.

\section*{Acknowledgments}
XG, JGC and MS thank support from the Basque Government BasQ initiative under the Q-STREAM project. They also acknowledge support from OpenSuperQ+100 (Grant No. 101113946) of the EU Flagship on Quantum Technologies, from Project Grant No. PID2024-156808NB-I00 and Spanish Ram\'on y Cajal Grant No. RYC-2020-030503-I funded by MI- CIU/AEI/10.13039/501100011033 and by “ERDF A way of making Europe” and “ERDF Invest in your Future”, and from the Spanish Ministry for Digital Transformation and of Civil Service of the Spanish Government through the QUANTUM ENIA project call Quantum Spain, by the EU through the Recovery, Transformation and Resilience Plan–Next Generation EU within the framework of the Digital Spain 2026 Agenda, and from Basque Government through Grant No. IT1887-26.


\section*{References}


\appendix

\section{Technical lemmas and proofs}
\label{section:proofs}
\begin{proof}[Proof of Lemma \ref{lemma:circ-eigvals-as-product}]
    The approach is similar to the construction of the quantum circuit for the QFT in \cite{Cleve_1998}.
    Each $k \in \mathbb{Z}_{2^n}$ determines 
    the binary encoding $k=\sum_{\ell=1}^n 2^{n-\ell}k_{\ell}$ with $k_{\ell} \in \{0, 1\}$, then from \eqref{eq:gamma-eig} we obtain
    \begin{subequations}
    \begin{align}
         \Gamma_n
         =& \sum_{k_1=0}^1 \sum_{k_2=0}^1 \cdots \sum_{k_n=0}^1
         \exp\left(\frac{2\pi \imath}{2^n}
         \sum_{\ell=1}^n \left(2^{n-\ell}k_{\ell}\right)
         \right)
         \ket{k_1}\bra{k_1} \otimes \cdots \otimes \ket{k_n}\bra{k_n}\\
         =& \sum_{k_1=0}^1 \sum_{k_2=0}^1 \cdots \sum_{k_n=0}^1
         \exp\left(\imath \pi k_{1}
         \right)
         \ket{k_1}\bra{k_1} \otimes \cdots \otimes
         \exp\left(\imath \pi 2^{1-n}k_{n}\right)
         \ket{k_n}\bra{k_n}\\
         =&
         \left(
         \ket{0}\bra{0} - \ket{1}\bra{1}
         \right) \otimes
         \sum_{k_2=0}^1 \cdots \sum_{k_n=0}^1
         \exp\left(\imath \pi 2^{-1}k_{2}\right)
         \ket{k_2}\bra{k_2} \otimes \cdots \otimes
         \exp\left(\imath \pi 2^{1-n}k_{n}\right)
         \ket{k_n}\bra{k_n}\\
         =&
         \pauliz{}
         \otimes
         \left(
         \ket{0}\bra{0} + e^{\imath \pi 2^{-1}}\ket{1}\bra{1}
         \right) \otimes
         \cdots
         \otimes
         \left(
         \ket{0}\bra{0} + e^{\imath \pi 2^{1-n}}\ket{1}\bra{1}
         \right),
    \end{align}
    \end{subequations}
    which corresponds to the claim in \eqref{eq:gamma-phase-gates}.
    
    This result can also be deduced from the fact that the first row of circulant matrices is the \underline{inverse} \emph{discrete Fourier transform} (DFT) of its eigenvalues \cite{ToeplitzReview}.
    So, the diagonal entries of $\Gamma_n$ correspond (up to a scalar factor) to
    \begin{align}
        \qftn{n}\left(\ket{0}^{\otimes(n-1)}\ket{1}\right) =& \frac{1}{\sqrt{2^n}}\bigotimes_{\ell=1}^n \left(\ket{0}+e^{2\pi \imath 2^{-\ell}}\ket{1}\right).
    \end{align}
\end{proof}
\section{Proof of Theorem \ref{thm:approx-circeigv-distance}}\label{app:approx-circeigv}

\begin{lemma}
    \label{lemma:fro-dist-uv-with-global-phase}
    Let $U, V$ be unitary operators on $n$ qubits, then
    \begin{align}
        \min_{\phi \in \mathbb{R}} \frac{\left\| U-e^{\imath \phi}V \right\|_F}{\sqrt{2^{n+1}}} =&
        \sqrt{1-\frac{|\traceop\left(U^{\dagger}V\right)|}{2^n}}.
    \end{align}
\end{lemma}
\begin{proof}
    We expand the cost function, so considering the unitary invariance of the Frobenius norm we obtain
    \begin{subequations}
    \begin{align}
        \frac{\left\| U-e^{\imath \phi}V \right\|_F}{\sqrt{2^{n+1}}} =&
        \frac{\left\| \idenmnodim-e^{\imath \phi}U^{\dagger}V \right\|_F}{\sqrt{2^{n+1}}} = \sqrt{1-\frac{1}{2^n}\repart\left(e^{\imath \phi}\traceop\left(U^{\dagger}V\right)\right)}.
    \end{align}
    \end{subequations}
    Now, let $\alpha e^{\imath \beta}=\traceop\left(U^{\dagger}V\right)$, then assuming $\alpha>0$ (i.e. $\traceop\left(U^{\dagger}V\right) \ne 0$), the expression
    $\repart\left(e^{\imath \phi} \traceop\left(U^{\dagger}V\right) \right)=\alpha \repart\left(e^{\imath (\phi + \beta)}\right)=\alpha \cos(\phi+\beta)$ is maximized when $\cos(\phi+\beta)=1$.
    Set $\phi=-\beta$, then $\repart\left(e^{\imath \phi} \traceop\left(U^{\dagger}V\right) \right)=\alpha$ with $\alpha=|\traceop\left(U^{\dagger}V\right)|$, therefore
    \begin{align}
       \min_{\phi \in \mathbb{R}} \sqrt{1-\frac{1}{2^n}\repart\left(e^{\imath \phi}\traceop\left(U^{\dagger}V\right)\right)}
       =&
       \sqrt{1-\frac{|\traceop\left(U^{\dagger}V\right)|}{2^n}},
    \end{align}
    which confirms the claim.
\end{proof}

\begin{proof}[Proof of Theorem \ref{thm:approx-circeigv-distance}]
    We first consider the case $k\ge 1$ and $k$ odd.
    By the half-angle formulae, for all $\theta \in \mathbb{R}$ s.t. $\cos(\theta/2)\ge 0$, we have that $|1+e^{\imath \theta}|=2\cos(\theta/2)$.
    In the upcoming steps we need that $\cos(\theta/2)\ge 0$ with $\theta=2\pi k/2^{n-d+\ell}$ for all $\ell \in [1,...,d]]$,
    that is $\pi k/2^{n-d+\ell}\le \pi/2$ so $k \le 2^{n-d}$ ($k$ positive).
    
    In addition, we consider the generalization of \emph{Morrie's law} \cite{morrie-law, math-bite}, which reads
    \begin{align}
        \label{eq:morries-law}
        \prod_{j=0}^{n-1}\cos(2^j \alpha) =& \frac{\sin(2^n \alpha)}{2^n \sin(\alpha)},
    \end{align}
    for all $\alpha \in \mathbb{R}$ and positive integers $n$.
    As a result of \ref{lemma:fro-dist-uv-with-global-phase}, for the case $1 \le d \le n-1$, we obtain
    \begin{subequations}
    \begin{align}
        \label{eq:approx-circeigv-distance-detail-i}
        \min_{\phi \in \mathbb{R}} \frac{1}{\sqrt{2^{n+1}}}\left\|\Gamma_{d,n}^k - e^{\imath \phi} \circeigv{n}^k \right\|_F =&
        \sqrt{1-\frac{
        \left|\traceop\left(
            \left(\Gamma_{d,n}^k\right)^{\dagger}\circeigv{n}^k
        \right)\right|
        }{2^n}}\\
        \underset{\eqref{eq:gamma-phase-gates}, \eqref{eq:gamma-approx}}{=}& 
        \sqrt{1-\frac{1}{2^d}\prod_{\ell=n-d+1}^n\left|1+\exp\left(\frac{2\pi \imath k}{2^{\ell}}\right)\right|}\\
        =&
        \sqrt{1-\prod_{\ell=1}^d\left|\frac{
            1+e^{{2\pi \imath k}/{2^{n-d+\ell}}}
        }{2}\right|}
        =
        \sqrt{1-\prod_{\ell=1}^d \cos\left(
        \frac{2\pi k}{2^{(n+1)+(\ell-d)}}
        \right)
        }\\
        =& 
        \sqrt{1-\prod_{\ell=0}^{d-1} \cos\left(
        \alpha 2^{\ell}
        \right)
        }
        \underset{\eqref{eq:morries-law}}{=}
        \sqrt{1-\frac{\sin(2^d \alpha)}{2^d \sin(\alpha)}},
    \end{align}
    \end{subequations}
    where $\alpha=\frac{2\pi k}{2^{n+1}}$.
    The second to the last equality holds since the following sets are equivalent
    $\left\{\frac{1}{2^{\ell-d}}\middle| \ell=1, 2, \ldots, d\right\}=\left\{2^{\ell}\middle| \ell=0, 2, \ldots, d-1 \right\}$.
    This confirms the claim for the case $1 \le d \le n-1$.
    However, when the approximation degree $d$ vanishes, so does the Frobenius norm above, therefore the latter result extends to the case $d=0$ as well.
    In addition, the approximation $\eta(\delta, k)$ follows at once from the small-angle approximation for $\sin(\alpha)$ when $n$ is large enough and $|k|\ll 2^n$.

    For the case $k<0$, the result follows directly by considering $k \leftarrow |k|$, indeed in \eqref{eq:approx-circeigv-distance-detail-i}
    we have that
    \begin{align}
        \left|\traceop\left(
            \left(\Gamma_{d,n}^k\right)^{\dagger}\circeigv{n}^k
        \right)\right|
        =&
        \left|
        \overline{
        \traceop\left(
            \left(\Gamma_{d,n}^k\right)^{\dagger}\circeigv{n}^k
        \right)}
        \right|
        =
        \left|\traceop\left(
            \left(\Gamma_{d,n}^{-k}\right)^{\dagger}\circeigv{n}^{-k}
        \right)\right|,
    \end{align}
    since $\overline{\circeigv{n}^k}=\circeigv{n}^{-k}$, for all (odd) integer $k$.
\end{proof}
\begin{corollary}
    Consider the conditions of Theorem \ref{thm:approx-circeigv-distance}, then
    \begin{align}
        \log(\eta(\delta, k)) \sim& -\delta \log\left(\frac{2}{\beta^{1/\delta}}\right),
    \end{align}
    with $\beta=\pi |k|$, as $\delta \to \infty$.
\end{corollary}
\begin{proof}
    Let $x=\beta/2^{\delta}$, then
    \begin{align}
        \lim_{\delta \to \infty} -
        \frac{\log(\eta(\delta, k))}{\left(\delta \log\left(\frac{2}{\beta^{1/\delta}}\right)\right)} =& \lim_{x\to 0} \frac{\log\left(1-\frac{\sin(x)}{x}\right)}{2\log(x)} 
        = \frac{1}{2} \lim_{x\to 0} \frac{\sin x + x\cos x}{\sin x} =1.
    \end{align}
\end{proof}

\section{Operator-norm error of the truncated eigenphases}
Theorem \ref{thm:approx-circeigv-distance} measures the truncation error in the normalised Frobenius norm, i.e.\ on average over the eigenphases.
To propagate the truncation through the Trotter bound of Theorem \ref{thm:qft-trotter} we need instead the operator norm $\norm{\cdot}$, which controls the worst eigenphase.
\begin{lemma}
    \label{lemma:approx-circeigv-opnorm}
    Let $n \ge 1$, $0 \le d \le n-1$, $\delta=n-d$, and let $k$ be any integer with $1 \le |k| \le 2^{\delta}$.
    Then, with $\alpha=2\pi|k|/2^{n+1}$ as in Theorem \ref{thm:approx-circeigv-distance},
    \begin{align}
        \label{eq:approx-circeigv-opnorm}
        \min_{\phi \in \mathbb{R}}\norm{\widetilde{\Gamma}_{d,n}^{k} - e^{\imath \phi}\circeigv{n}^{k}} \le
        \frac{\pi \abs{k}}{2^\delta}.
    \end{align}
\end{lemma}
\begin{proof}
    Write the computational basis index as $j=2^d q + r$ with $0\le q<2^{n-d}$ and $0 \le r < 2^d$.
    From Eqs.~\eqref{eq:gamma-eig} and \eqref{eq:gamma-approx}, we have that
    \begin{align}
        \begin{split}
        \circeigv{n}^k\ket{j} =& \omega^{kj}\ket{j},\\
        \widetilde{\Gamma}_{d,n}^k\ket{j} =& \omega^{k2^dq}\ket{j}=\omega^{k(j-r)}\ket{j},
        \end{split}
    \end{align}
    hence
    \begin{align}
        \widetilde{\Gamma}_{d,n}^{-k}\circeigv{n}^{k} =& \idenm{2}^{\otimes (n-d)} \otimes \sum_{r=0}^{2^d-1}\omega^{kr}\ket{r}\bra{r}.
    \end{align}
    By the unitary invariance of the operator norm,
    \begin{align}
        \norm{\widetilde{\Gamma}_{d,n}^{k} - e^{\imath \phi}\circeigv{n}^{k}} =&
        \norm{\idenmnodim - e^{\imath\phi}\widetilde{\Gamma}_{d,n}^{-k}\circeigv{n}^{k}}\\
        =& \max_{0 \le r < 2^d}\abs{1-e^{\imath(\phi+2\pi k r/2^n)}}\\
        =& 2\max_{0\leq r \leq 2^d}\abs{\sin(\phi + 2\pi k r/2^n)}.
    \end{align}
We can see that for $r \in [0,2^d]$ $2\pi k/2^n \leq \pi/2$, where $\abs{\sin(\cdot)}$ increases monotonically. Thus,
\begin{align}
    \max_{0\leq r \leq 2^d}\abs{\sin(\phi + 2\pi k r/2^n)} = \max\{\abs{\sin(\phi)}, \abs{\sin(\phi + (2^d-1)2\pi k/2^n)}\}.
\end{align}
This quantity is minimized for $\phi$ being equidistant from both endpoints, $\phi^\star = -(2^d-1)\pi k/2^{n+1}$. We can thus write
\begin{align}
     \min_{\phi \in \mathbb{R}}\norm{\widetilde{\Gamma}_{d,n}^{k} - e^{\imath \phi}\circeigv{n}^{k}} \leq 2\abs{\sin\left( \frac{(2^d-1)\pi k}{2^{n+1}}\right)} \leq \frac{(2^{d}-1)\pi \abs{k}}{2^n}\leq \frac{2^d\pi \abs{k}}{2^n} = \frac{\pi \abs{k}}{2^\delta}.
\end{align}
\end{proof}

\section{Proof of Theorem \ref{thm:overlap-zi-gamma-pows}}\label{app:overlap}

\begin{lemma}
    \label{lemma:cot-i}
    Let $n\ge 2$, $\ell_1 \ge 2$ and $k$ be integers with $k$ \underline{odd} and $n \ge \ell_1$.
    Let $\omega=\exp(2\pi \imath\, k/2^n)$ and $\xi=\omega^{2^{n-\ell_1+1}}=\exp(2\pi \imath \,k/2^{\ell_1-1})$.
    Then
    \begin{subequations}
    \begin{align}
        \label{eq:cot-i}
        \prod_{\ell=\ell_1}^n\left(1 + \omega^{2^{n-\ell}}\right)
        =&
        \prod_{\ell=\ell_1}^n\left(1 + \exp\left(\frac{2\pi \imath k}{2^{\ell}}\right)\right)\\
        =&
        \frac{1-\xi}{2}
        \left(1 + \imath \cot\left(\frac{\pi k}{2^n}\right)\right)
        \sim
        \imath\frac{1-\xi}{2\pi}
        \cdot
        \frac{2^n}{k},
    \end{align}
    \end{subequations}
    where the asymptotic equivalence is w.r.t. $n\to \infty$.
    As special case, we have the factor $(1-\xi)/2=1$ when $\ell_1=2$.
\end{lemma}
\begin{proof}
    Since $\gcd(2, k)=1$ and $n\ge 2$, then $\omega\ne 1$.
    Note that
    \begin{align}
        (1-\omega)\prod_{\ell=\ell_1}^n \left(1+\omega^{2^{n-\ell}}\right) =& (1-\omega^2)\prod_{\ell=\ell_1}^{n-1} \left(1+\omega^{2^{n-\ell}}\right),
    \end{align}
    which recursively reduces to
    \begin{align}
        (1-\omega)\prod_{\ell=\ell_1}^n \left(1+\omega^{2^{n-\ell}}\right) =&
        \left(1-\omega^{2^{n-\ell_1}}\right)\left(1+\omega^{2^{n-\ell_1}}\right)=1-\xi.
    \end{align}
    When $\ell_1=2$, we have the special case in with the latter reduces to $(1-\imath^k)(1+\imath^k)=1-\xi=2$,
    since $2\nmid k$.
    So we obtain that
    \begin{align}
        \label{eq:cot-i-proof-detail-i}
        \prod_{\ell=\ell_1}^n \left(1+\omega^{2^{n-\ell}}\right) =&
        \frac{1-\xi}{2}
        \cdot
        \frac{2}{1-\omega},
    \end{align}
    which is well defined since $\omega\ne 1$.

    Now, since $(1-\omega)(1-\overline{\omega})=2(1-\repart(\omega))$, we have that
    \begin{subequations}
    \begin{align}
        \frac{2}{1-\omega} =& 
        \frac{2(1-\overline{\omega})}{(1-\omega)(1-\overline{\omega})}\\
        =& \frac{1-\overline{\omega}}{1-\repart(\omega)}\\
        =& \frac{1-\repart(\omega)+\imath \impart(\omega)}{1-\repart(\omega)}\\
        =& 1+ \imath\frac{\sin(2\pi k/2^n)}{1-\cos(2\pi k/2^n)}\\
        \label{eq:cot-i-proof-detail-ii}
        =& 1 + \imath \cot\left(\frac{\pi k}{2^n}\right),
    \end{align}
    \end{subequations}
    with the latest equality following from the half-angle formula $\cot(\theta/2)=\frac{\sin(\theta)}{1-\cos(\theta)}$.
    Hence, the second equality of claim in \eqref{eq:cot-i}, is a consequence of \eqref{eq:cot-i-proof-detail-i} and \eqref{eq:cot-i-proof-detail-ii}.

    Finally, since the mapping $z \mapsto \pi \cot(\pi z)$ has simple poles for $z\in \mathbb{Z}$ with residue $1$ \cite{Lang1999},
    then the asymptotic relation follows directly (noting that $\xi$ effectively does not depend upon $n$).
\end{proof}

\begin{lemma}
    \label{lemma:trace-herm-skewh}
    For any conformable complex square matrices $A$ and $B$ with $B$ Hermitian and $[A^{\dagger}, B]=0$, then
    \begin{subequations}
    \begin{align}
        \label{eq:trace-herm-skewh}
        \traceop\left(\hermpart(A)\right) = \repart \traceop(A),\quad&
        \traceop\left(\skewhpart(A)\right) = \imath\impart \traceop(A),\\
        \label{eq:herm-skewh-rel-with-mul}
        B \hermpart(A) = \hermpart(BA),\quad& B \skewhpart(A) = \skewhpart(BA).
    \end{align}
    \end{subequations}
\end{lemma}

\begin{proof}[Proof of Theorem \ref{thm:overlap-zi-gamma-pows}]
    Since $O_{\boldsymbol{p}}$ is Hermitian and $[O_{\boldsymbol{p}}, \circeigv{n}^{\pm k}]=0$ (required for \ref{lemma:trace-herm-skewh}), then
    \begin{subequations}
    \begin{align}
        \traceop\left(O_{\boldsymbol{p}}
            \frac{\Gamma^k_n + (-1)^s \Gamma_n^{-k}}{2}
        \right) =&
        \frac{1+(-1)^s}{2}\traceop\left(O_{\boldsymbol{p}}\hermpart\left(\circeigv{n}^k\right)\right)
        + \frac{1-(-1)^s}{2}\traceop\left(O_{\boldsymbol{p}} \skewhpart\left(\circeigv{n}^k\right)\right)\\
        \underset{\eqref{eq:herm-skewh-rel-with-mul}, \eqref{eq:trace-herm-skewh}}{=}&
        \frac{1+(-1)^s}{2}\repart\traceop\left(O_{\boldsymbol{p}} \circeigv{n}^k\right)
        + \imath\frac{1-(-1)^s}{2}\impart\traceop\left(O_{\boldsymbol{p}} \circeigv{n}^k\right).
    \end{align}
    \end{subequations}
    Consider the arguments of real and imaginary part, so
    \begin{align}
        \traceop\left(O_{\boldsymbol{p}} \circeigv{n}^k\right) =&
        \traceop\left(R^k\right)
        \prod_{\ell=1}^{n-d} \traceop\left(\pauliz^{p_{\ell}} P\left(2\pi k / 2^{\ell}\right)\right)\\
        =& 
        \traceop\left(R^k\right)
        \prod_{\ell=1}^{n-d}\left(1 + (-1)^{p_{\ell}}\exp\left(\frac{2\pi \imath k}{2^{\ell}}\right)\right),
    \end{align}
    with $R^k=\bigotimes_{\ell=n-d+1}^n P(2\pi k / 2^{\ell})$.
    In the claim we consider the case where $k\ge 1$ is an odd integer. 
    So, as a result of \ref{lemma:cot-i} (with $\ell_1=n-d+1\ge 2$, $n\ge 2$ and $k$ odd) we have that
    \begin{align}
        \traceop\left(R^k\right) =& \prod_{\ell=n-d+1}^n \left(1+e^{2\pi \imath k/2^{\ell}}\right)\\
        =&
        \frac{1-\xi_k}{2}
        \left(1 + \imath \cot\left(\frac{\pi k}{2^n}\right)\right),
    \end{align}
    with $\xi_k=\exp(2\pi \imath \,k/2^{n-d})$.
\end{proof}

\begin{corollary}
    \label{corollary:overlap-zi-gamma-pows-tr-scale}
    Under the conditions of Theorem \ref{thm:overlap-zi-gamma-pows}, 
    assume $\frac{k}{2^n}=\bigoh\left(\frac{1}{2^n}\right)$, then
    \begin{subequations}
    \begin{align}
        \frac{\left|\traceop\left(O_{\boldsymbol{p}} \circeigv{n}^k\right)\right|}{2^n}=&
        \frac{1}{2^n}
        \left|
            \frac{1-\xi_k}{2}
            \left(1 + \imath \cot\left(\frac{\pi k}{2^n}\right)\right)
        \right|\cdot
        \prod_{\ell=1}^{n-d}\left|1 + (-1)^{p_{\ell}}\exp\left(\frac{2\pi \imath k}{2^{\ell}}\right)\right|\\
        =&
        \sin\left(\frac{\pi k}{2^{n-d}}\right)
        \cdot
        \frac{1}{2^n \sin(\pi k/2^n)}
        \cdot
        \prod_{\ell=1}^{n-d}2 \sqrt{\frac{1 + (-1)^{p_{\ell}}\cos(2\pi k/2^{\ell})}{2}}\\
        =&
        2^{n-d}\sin\left(\frac{\pi k}{2^{n-d}}\right)
        \cdot
        \frac{\bigoh(2^n)}{2^n}
        \cdot
        \prod_{\ell=1}^{n-d}\sqrt{\frac{1 + (-1)^{p_{\ell}}\cos(2\pi k/2^{\ell})}{2}}\\
        =&
        \gamma
        2^{n-d}\sin\left(\frac{\pi k}{2^{n-d}}\right)
        \cdot
        \prod_{\ell=1}^{n-d}\sqrt{\frac{1 + (-1)^{p_{\ell}}\cos(2\pi k/2^{\ell})}{2}},
    \end{align}
    \end{subequations}
    for some constant $\gamma > 0$.
\end{corollary}

\section{Unitary diagonalisation of the skew-circulant generator}\label{app skew-circulant diag}

Using the eigenvalue matrix $\Gamma_n$ of the cyclic generator [Eq.~\eqref{eq:gamma-eig}] and the QFT operator [Eq.~\eqref{qftn operator}], we establish the following lemma.
\begin{lemma}
    \label{lemma:multictrl-qft}
    Let $t\in [0, 1)$ and let $\omega=\exp(2\pi \imath/N)$ where $N=2^n$ given the positive integer $n$.
    Then,
    \begin{align}
        \label{eq:multictrl-qft}
        \omega^t \circeigv{n}^{t} \left(C_n \circeigv{n}^{-t} C_n^{\top}\right) =& 
        \idenmnodim + \left(e^{2\pi \imath\, t} - 1\right) \ket{N-1}\bra{N-1}.
    \end{align}
\end{lemma}
\begin{proof}
    The claim follows directly from the action of the circulants on the diagonals $\circeigv{n}^{\pm t}$.
   We note that $C_n\ket{k} = \ket{(k-1) \mod N}$ and $\bra{k}C_n^\top = \bra{(k-1) \mod N}$.
    Then,
    \begin{subequations}
    \begin{align}
        \omega^t \circeigv{n}^{t} \left(C_n \circeigv{n}^{-t} C_n^{\top}\right) =& \omega^t \sum_{k=0}^{N-1}w^{tk}\ket{k}\bra{k}\sum_{k^\prime=0}^{N-1}w^{tk^\prime}C_n\ket{k^\prime}\bra{k^\prime}C_n^\top = \n
        =&
        \omega^t \sum_{k=0}^{N-2} \omega^{tk} \omega^{-t(k+1)} \ket{k}\bra{k}
        +
        \omega^{t} \omega^{t(N-1)} \ket{N-1}\bra{N-1}\\
        =&
        \left(\idenmnodim - \ket{N-1}\bra{N-1}\right) + \omega^{tN} \ket{N-1}\bra{N-1},
    \end{align}
    \end{subequations}
    but $\omega^{tN}=e^{2\pi \imath\, t}$, hence the claim is proved.
\end{proof}
An important corollary of this lemma is found by multiplying the left hand side expression in Eq.(\ref{eq:multictrl-qft}) by $C_n$, and noticing it corresponds to the definition of the skew-circulant generator in Eq.(\ref{skew-circulant generator}). 
\begin{corollary}
    Consider the skew-circulant generator $N_n$ defined in Eq.(\ref{skew-circulant generator}). Then
    \begin{align}
        N_n= \sqrt{\omega}\sqrt{\Gamma_n}C_n\sqrt{\Gamma_n}^\dagger.
    \end{align}
\end{corollary}
Thus, the eigenvectors of $N_n$ are given by $\sqrt{\Gamma_n}\qftn{n}$, with corresponding eigenvalues $\sqrt{\omega}\Gamma_n$.
\section{Banded Toeplitz operator as combination of diagonalisable matrices}\label{app:banded-toeplitz}
\begin{proof}[ Proof of Lemma \ref{lemma:banded-T-D1-D2}]
Starting from Eq.(\ref{eq:sum-Cn-Ln}), when $k=0, 1, \ldots, m$ with $m < 2^n$ by assumption, we have that
    \begin{align}
        \left(L_n^{\top}\right)^k = \frac{C_n^k + N_n^k}{2},
        \quad&
        L_n^k = \frac{C_n^{-k} + N_n^{-k}}{2},
    \end{align}
    then together with Eqss(\ref{eq:Cn-diag}) 
    and (\ref{eq:Nn-diag}), the expression in \eqref{eq:toeplitz-poly} becomes
    \begin{subequations}
    \begin{align}
        T =&
        \alpha_0 \idenmnodim + \sum_{k=1}^{m} \left(
            \alpha_{-k} L_n^k +
            \alpha_{k} \left(L_n^{\top}\right)^k
        \right)\\
        =&
        \alpha_0 \idenmnodim + \sum_{k=1}^{m} \left(
            \alpha_{-k} \frac{C_n^{-k} + N_n^{-k}}{2} +
            \alpha_{k} \frac{C_n^k + N_n^k}{2}
        \right)\\
        =&
        \label{eq:new-trotter-lemma-proof-i}
        \frac{1}{2}\sum_{k=-m}^{m} \alpha_{k} \left(C_n^k + N_n^k\right)\\
        =&
        \frac{1}{2}\sum_{k=-m}^{m} \alpha_{k} \left(
            \qftn{n}\circeigv{n}^k \qftdgn{n} +
            \sqrt{\circeigv{n}}\qftn{n}
                \left(\omega^{\frac{k}{2}} \circeigv{n}^k\right)
            \qftdgn{n}\sqrt{\circeigv{n}}^{\dagger}
        \right),
    \end{align}
    \end{subequations}
    which corresponds to the claim in \eqref{eq:new-trotter-lemma-first}.
\end{proof}

\section{QFT-Trotter product formula}\label{app:qft-trotter}
\begin{proof}[Proof of Theorem \ref{thm:qft-trotter}]

Let $\widetilde{A} = \qftn{n}\exp\left(\widetilde{D}_1/2\right)\qftn{n}^\dagger$, with $\widetilde{D}_1 = \sum_{k=-m}^m\alpha_k \widetilde{\Gamma}_{d,n}^k$, and \\$\widetilde{B}=\sqrt{\Gamma_n}\qftn{n}\exp\left( \widetilde{D}_2/2\right)\qftn{n}^\dagger\sqrt{\Gamma_n}^\dagger$, with $\widetilde{D}_2 = \sum_{k=-m}^m\alpha_k \omega^{k/2}\widetilde{\Gamma}_{d,n}^k$. The Trotter error is then
\begin{align}\label{eq:complete-trotter-error}
    \norm{e^{Tt}-\left(\widetilde{A}^{t/c}\widetilde{B}^{t/c}\right)^c}\leq \norm{e^{Tt}-\left(A^{t/c}B^{t/c}\right)^c} + \norm{\left(A^{t/c}B^{t/c}\right)^c-\left(\widetilde{A}^{t/c}\widetilde{B}^{t/c}\right)^c}.
\end{align}

The first term can be bounded by direct application of the standard Trotter-error estimate bound \cite{Hall2015} we have for large $c$
\begin{align}
   \norm{e^{Tt} - \left( A^{t/c}B^{t/c}\right)^c} = \mathcal{O}\left(\exp\left[t\left(\max_j \repart D_{1,jj} + \max_j \repart D_{2,jj}\right)\right]\frac{t^2}{c}\norm{[\log(A),\log(B)]} \right),
\end{align}
where we have used that $\log(A)$ and $\log(B)$ are normal operators, since they can be written as a polynomial on the normal operators $C_n$ adn $N_n$, respectively, and thus $\norm{e^{t \log(A/B)}} = \abs{e^{t\max_j  \lambda_j^{A/B}}} = e^{t \max_j \repart D_{1/2,jj}}$. If we assume that $\max_j \repart D_{1/2,jj}\leq 0$, then $e^{t \max_j \repart D_{1/2,jj}} \leq e^0 = 1$.

Using the diagonalisation of $C_n$, Eq.(\ref{eq:Cn-diag}), of $N_n$ (\ref{eq:Nn-diag}) and the definition of the operators $A$ and $B$ in Eq.(\ref{eq:A-B}) we see that
\begin{align}
\log(A) = \frac{1}{2}\sum_{k=-m}^m\alpha_k C_n^k,\q \log(B) = \frac{1}{2}\sum_{k=-m}^m\alpha_k N_n^k,
\end{align}
and thus
\begin{align}\label{eq:logA-logB}
\norm{[\log(A),\log(B)]} \leq \frac{1}{4}\sum_{k=-m}^m\sum_{l=-m}^m \abs{\alpha_k\alpha_l}\norm{[C_n^k,N_n^l]}\leq \frac{1}{2}\sum_{k=-m}^m\sum_{l=-m}^m \abs{\alpha_k\alpha_l},
\end{align}
where we have used that $C_n^k$ and $N_n^l$ are unitaries, and for any pair of unitaries $U,V$ we have that $\norm{[U,V]} \leq \norm{UV}+\norm{VU}\leq 2$. The $1/2$ factor can be absorved by $\mathcal{O}(\cdot)$. Then, from Eq.(\ref{eq:logA-logB}) we find
\begin{align}
\sum_{k=-m}^m\sum_{l=-m}^m \abs{\alpha_k\alpha_l} =\left(\sum_{k=-m}^m\abs{\alpha_k}\right)^2 \equiv \norm{\alpha}_1^2.
\end{align}

The second term in Eq.(\ref{eq:complete-trotter-error}) can be bounded using that
\begin{align}
    \norm{\left(A^{t/c}B^{t/c}\right)^c-\left(\widetilde{A}^{t/c}\widetilde{B}^{t/c}\right)^c} \leq c \norm{A^{t/c}B^{t/c}-\widetilde{A}^{t/c}\widetilde{B}^{t/c}} \leq c\left(\norm{A^{t/c}-\widetilde{A}^{t/c}} + \norm{B^{t/c}-\widetilde{B}^{t/c}} \right),
\end{align}
where we have used that, since $\repart{D_{1/2}} \leq 0$, $\norm{A^{t/c}} \leq 1$ and $\norm{B^{t/c}}\leq 1$. Because the unitary $\qftn{n}$ preserves the norm, we have that
\begin{align}
    \norm{A^{t/c}-\widetilde{A}^{t/c}} \leq& \norm{\exp\left(D_1t/2c\right)-\exp\left(\widetilde{D}_1t/2c\right)} \leq \frac{t}{2c}\norm{D_1-\widetilde{D}_1} \leq \n
    &\frac{t}{2c}\sum_{k=-m}^m\abs{\alpha_k}\norm{\Gamma_n^k-\widetilde{\Gamma}_{d,n}^k} \leq \frac{t}{2c}\sum_{k=-m}^m\abs{\alpha_k}\frac{\pi \abs{k}}{2^\delta} \leq \frac{\pi tm}{2^{\delta+1}c}\norm{\alpha}_1,
\end{align}
where we have used the result of Appendix \ref{app:overlap}. The same can be done for the term $\norm{B^{t/c}-\widetilde{B}^{t/c}}$, such that the final bound is
\begin{align}
    \norm{\left(A^{t/c}B^{t/c}\right)^c-\left(\widetilde{A}^{t/c}\widetilde{B}^{t/c}\right)^c} \leq \mathcal{O}\left( \frac{\pi t m}{2^\delta}\norm{\alpha}_1\right).
\end{align}
\end{proof}

\section{Discrete Cosine/Sine Transformations}\label{app DCT DST}
\begin{proof}[Proof of Eq.(\ref{eq:DCT-DST-diag})]
    \begin{align}
        \qftn{n+1}\Delta_p^{(n+1)} \qftn{n+1}^\dagger = \Lambda_{n+1} \to  \Delta_p^{(n+1)} = \qftn{n+1}^\dagger \Lambda_{n+1} \qftn{n+1},
    \end{align}
    where $\Lambda_{n+1} = $diag$(\lambda_k)$ acting on $n+1$ qubits. We have the relation \cite{DCT-quantum}
    \begin{align}
        U_{n+1}^\dagger \qftn{n+1}V_{n+1}
 = C_N^{II}\oplus (-i)S_N^{II}, 
 \end{align}
 and thus
 \begin{align}
     V_{n+1} = \qftn{n+1}^\dagger U_n(C_{II} \oplus -iS_{II}),\n
     V_{n+1}^\dagger = (C_{II}\oplus iS_{II})U_{n+1}^\dagger \qftn{n+1},
 \end{align}
 such that
 \begin{align}
     &V_{n+1}^\dagger \Delta_p^{(n+1)} V_{n+1} = V_{n+1} \qftn{n+1}^\dagger \Lambda_{n+1}\qftn{n+1} V_{n+1}^\dagger =\n
     =& (C_{II}\oplus iS_{II}) U_{n+1}^\dagger \qftn{n+1} \qftn{n+1}^\dagger \Lambda_{n+1} \qftn{n+1}\qftn{n+1}^\dagger U_{n+1} (C_{II}\oplus (-i)S_{II})=\n
     =&(C_{II}\oplus iS_{II})U_{n+1}^\dagger \Lambda_{n+1} U_{n+1} (C_{II}\oplus (-i)S_{II}).
 \end{align}

From the definition of $U_{n+1}$ in \cite{DCT-quantum} we can check that
\begin{align}
    U_{n+1}^\dagger \Lambda_{n+1} U_{n+1} = \Lambda_n^{(N)} \oplus \Lambda_n^{(N)},
\end{align}
where $\Lambda_n^{(N)}$ is the diagonal matrix with the eigenvalues of the Neumann boundary condition Laplacian acting on $n+1$ qubits. These have the same values as the eigenvalues of the periodic Laplacian acting on $n$-qubits. Finally, since $C_{II}$ is the matrix diagonalizing the Neumann Laplacian, we obtain
\begin{align}
\left(C_{II}\oplus iS_{II}\right)\left(\Lambda_n^{(N)} \oplus \Lambda_n^{(N)}\right) \left(C_{II}\oplus (-i)S_{II}\right) = \left(C_{II}\Lambda_n^{(N)}C_{II} \right)\oplus \left(iS_{II}\Lambda_n^{(N)}(-i)S_{II} \right) = \Delta_N^{(n)}\oplus \Delta_D^{\prime (n)},
\end{align}
where $\Delta_D^{\prime (n)}$ is the Dirichlet Laplacian with $-3$ in the $(1,1)$ and $(N,N)$ entries.
\end{proof}

\end{document}